\documentclass[journal]{IEEEtran}

\usepackage[dvips]{graphicx}
\usepackage[cmex10]{amsmath}

\usepackage{color}
\usepackage{amsthm}
\usepackage{amssymb}
\usepackage{cases}
 
\newcommand{\bra}[1]{\langle {#1} |}
\newcommand{\ket}[1]{| {#1} \rangle}

\newtheorem{lemma}{Lemma}
\newtheorem{definition}{Definition}
\newtheorem{theorem}{Theorem}

\begin{document}

\title{Network coding for distributed quantum computation over cluster and butterfly networks}

\author{Seiseki~Akibue~
        and~Mio~Murao
\thanks{S. Akibue is with Graduate School of Science, The University of Tokyo, Tokyo 113-0033 Japan and present affiliation of S. Akibue is NTT Communication Science Laboratories, NTT Corporation, Kanagawa 243-0198 Japan, e-mail: akibue.seiseki@lab.ntt.co.jp.}
\thanks{M. Murao is with Graduate School of Science, The University of Tokyo, Tokyo 113-0033 Japan and NanoQuine, The University of Tokyo, Tokyo 113-0033 Japan, email: murao@phys.s.u-tokyo.ac.jp.}
}

\date{\today}
%

\maketitle

\begin{abstract}
To apply {\it network coding for quantum computation}, we study the distributed implementation of unitary operations over all separated input and output nodes of quantum networks. We consider  networks where quantum communication between nodes is restricted to sending a qubit, but classical communication is unrestricted.  We analyze which $N$-qubit unitary operations are implementable over {\it cluster networks} by investigating transformations of a given cluster network into quantum circuits.  We show that any two-qubit unitary operation is implementable over the {\it butterfly network} and the {\it grail network},  which are fundamental primitive networks for classical network coding.  We also analyze probabilistic implementations of unitary operations over cluster networks.
\end{abstract}


%
\IEEEpeerreviewmaketitle

\section{Introduction}
%
%
%
%
\IEEEPARstart{D}{istributed} quantum computation is computation over a network consisting of spatially separated quantum systems represented by nodes connected by mediating quantum systems represented by edges.   A serious problem for any kind of distributed computation is the \textit{bottleneck} problem  caused by the collision of communication pathways between the nodes.  The bottleneck problem worsens as the network grows. Thus it is important to consider how to optimize transmission protocols so that the amount of quantum communications is reduced.  

In classical network information theory, \textit{network coding}, which incorporates processing at each node in addition to routing, provides efficient transmission protocols that can resolve the bottleneck problem \cite{Ahlswede}.   As an example, consider a communication task over the {\it butterfly network} and the {\it grail network} presented in Fig.~\ref{fig:cgrail} that aims to transmit single bits $x$ and $y$ from $i_1$ to $o_2$ and $i_2$ to $o_1$ simultaneously via nodes $n_1$, $n_2$, $n_3$ and $n_4$. The directed edges denote transmission channels with $1$-bit capacity. One of the channels in each network (the channel from $n_1$ to $n_2$ for the butterfly network and either the channel from $n_1$ to $n_2$ or the channel from $n_3$ to $n_4$ for the grail network) exhibits the bottleneck without network coding shown in Fig.~\ref{fig:cgrail}.

\begin{figure}
 \begin{center}
  \includegraphics[height=.19\textheight]{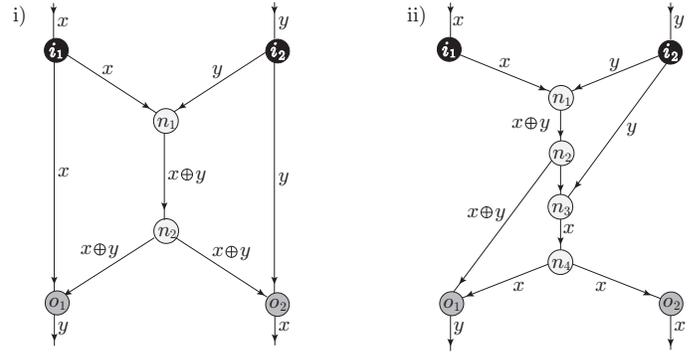}
  \end{center}
  \caption{Network coding for a classical communication task over i) the butterfly network and ii) the grail network. Two bits of information $x, y \in \{0, 1\}$ are given at the input nodes $i_1$ and $i_2$, respectively. $x \oplus y$ denotes addition of $x$ and $y$ modulus $2$.}
\label{fig:cgrail}
\end{figure}

{\it Quantum} communication with \textit{quantum network coding} has been studied by analogy to classical network coding \cite{Iwama, Leung, Hayashi,Kobayashi0, Kobayashi1, Kobayashi2}.   $k$-pair quantum communication over a network is a unicast communication task to faithfully transmit a $k$-qubit state given at distinct input nodes $\{i_1, i_2,\cdots,i_k\}$ to distinct output nodes $\{o_{1},o_{2},\cdots,o_k\}$ through a given network.  Two examples of $2$-pair quantum communication over a butterfly network and a grail network are shown in Fig.~\ref{fig:qgrail}.

In quantum mechanics, the no-cloning theorem forbids the creation of a perfect copy of an unknown state.  Thus perfect multicast communication of an unknown input state is impossible.    As copying states is a key element of classical network coding, classical network coding cannot be simply extended to $k$-pair quantum communication over the networks.   Indeed, in the setting where each edge can be used for either 1-bit classical communication or 1-qubit quantum communication, perfect quantum 2-pair communication over the butterfly network has been shown to be impossible \cite{Iwama, Leung}. However, it has been shown that if each edge can be used for either 2-bit classical communication or 1-qubit quantum communication, perfect quantum 2-pair communication over the butterfly network is possible, if and only if input nodes share two Bell pairs \cite{Hayashi}. 

 Further, if each edge has 1-qubit channel capacity and classical communication is freely allowed between any nodes, it has been shown that there exists a quantum network coding protocol to achieve the 2-pair quantum communication over the butterfly and grail networks perfectly \cite{Kobayashi0, Kobayashi1, Kobayashi2}.  This setting is justified in practical situations, where classical communication is much easier to implement than quantum communication.  Moreover in this setting, the correspondence between classical and quantum network coding has been discovered.   Namely,  it has been shown that $k$-pair quantum communication is possible over a network if  the corresponding $k$-pair classical communication is possible over the network using linear classical network coding schemes \cite{Kobayashi0, Kobayashi1} or even using nonlinear schemes \cite{Kobayashi2}. A connection between linear classical network coding schemes and measurement-based quantum computation has been investigated in \cite{Roetteler}.  However, it has been an open problem to determine the possibility of $k$-pair quantum communication over networks where the corresponding $k$-pair classical communication is impossible.

\begin{figure}
 \begin{center}
  \includegraphics[height=.16\textheight]{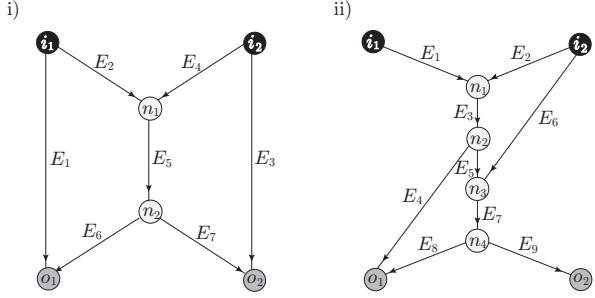}
  \end{center}
  \caption{i) The butterfly network and ii) the grail network with the input nodes ($i_1$ and $i_2$), output nodes ($o_1$ and $o_2$) and the repeater nodes ($n_1,n_2,n_3$ and $n_4$). The directed edges $E_1, E_2,\cdots,E_9$ represent quantum or classical channels. Quantum channels have 1-qubit capacity.  There are several settings of classical channel capacity for each edge.   Our task is to transmit a given two-qubit state $|{\rm input}\rangle_{i_1,i_2}$ from $i_1$ to $o_2$ and from $i_2$ to $o_1$ simultaneously by using the channels and local quantum operations at each nodes under the setting where classical communication is freely allowed between edges.}
\label{fig:qgrail}
\end{figure}

In $k$-pair quantum communication, the output state $\left \vert {\rm output}\right \rangle_{o_{1}\cdots o_k}$ can be regarded as a state obtained by performing a $k$-qubit unitary operation $U$ on the input state  $\left \vert {\rm input}\right \rangle_{i_1\cdots i_k}$
\begin{equation}
\label{eq:quantumcoding}
\left \vert {\rm output}\right \rangle_{o_1\cdots o_k} = U \left \vert {\rm input}\right \rangle_{i_1\cdots i_k},
\end{equation}
where $U$ is a permutation operation. 
We do not need to restrict the $k$-qubit unitary operation $U$ in Eq.\eqref{eq:quantumcoding} to be a permutation operation, it can be a general quantum operation.   This leads to the idea of \textit{network coding for quantum computation}, which aims to perform a quantum operation on a state given at distinct input nodes and to faithfully transmit the resulting state to the distinct output nodes efficiently over the network at the same time. By computing and communicating simultaneously, quantum computation over the network may reduce communication resources in the distributed quantum computation scenario.   The study of network coding for computation is still in its infancy for classical and quantum cases.   Network coding for classical computation is considered in  \cite{NetworkCodingforComputation} and network coding for quantum computation over the butterfly network is considered in \cite{Soeda}, both in 2011.

In this paper,  we investigate a {\it cluster network}, which is a special class of network with $k$ input nodes and $k$ output nodes, as a first step to apply network coding for more general quantum computation. The cluster network contains the grail network as its special case.  We focus on the setting considered in \cite{Kobayashi0, Kobayashi1, Kobayashi2}, where classical communication is freely allowed between any two nodes.   We identify the class of unitary operations that can be implemented over cluster networks in this setting by investigating transformations of cluster networks into quantum circuits implementable by using quantum communication resources between nodes specified by the cluster network.   The transformation method of cluster networks is also applicable to the butterfly network.   It provides constructions of quantum network coding for implementing any two-qubit unitary operations over the grail and butterfly networks, which are  the fundamental primitive networks for classical network coding.  We also analyze probabilistic implementation of $N$-qubit unitary operations over the cluster network to understand the properties of quantum network coding for quantum computation when the requirement of deterministic implementations are relaxed but that of exact implementations are kept.  

The rest of this paper is organized as follows. In Section II, we review necessary background information. In Section III, we define the cluster network and the implementability of a unitary operation over a quantum network.  In Section IV, we present a method to convert a given cluster network into quantum circuits describing unitary operations that are implementable over the network.  In Section V, we show that any two-qubit unitary operations is implementable over butterfly and grail networks.   In Section VI, we investigate the condition for unitary operations to be implementable over a given cluster network and show that our conversion method presented in Section IV gives all implementable unitary operations over the cluster networks with 2 and 3 input nodes.  Since the condition for unitary operations presented in Section VI is not based on classical network coding schemes but based on properties of quantum operations, it can be used for analyzing $k$-pair quantum communication over networks where corresponding $k$-pair classical communication is impossible.   In Section VII, we investigate probabilistic implementation of unitary operations over the cluster network and show that it is impossible to achieve $2$-pair quantum communication even probabilistically over a square shaped four-node quantum network, where corresponding $2$-pair classical communication is impossible.  A conclusion is given in Section VIII.

\section{Preliminaries}
\subsection{Notations}
The following notation will be used throughout this paper.

\begin{table}[htb]
\begin{tabular}{ll}
$\overline{a}$&
The complex conjugate of $a$.\\
$a^T$&
The transpose of $a$.\\
$a^{\dag}$&
The conjugate transpose of $a$.\\
$\mathbf{L}(\mathcal{H})$&
The set of linear operators \\&acting on the Hilbert space $\mathcal{H}$. \\
$\mathbb{I}_A$&
The identity operator on $\mathcal{H}_A$. \\
$\mathbf{U}(\mathcal{H})$&
The set of unitary operators.\\
$\mathbf{U}_c$&
The set of unitary operators \\&locally unitarily equivalent to\\& a two-qubit controlled unitary operation.\\
$\rm{tr}$&
The trace of a linear operator.\\
$\mathbf{L}(\mathcal{H}_A:\mathcal{H}_B)$&
The set of linear operators. \\&$\mathbf{L}:\mathcal{H}_A\rightarrow\mathcal{H}_B$.\\
$\mathbf{U}(\mathcal{H}_A:\mathcal{H}_B)$&
The set of isometry operators.\\
$\textsc{Sch\#}_B^A(\ket{\psi})$&
The Schmidt rank of $\ket{\psi}\in\mathcal{H}_A\otimes\mathcal{H}_B$.\\
$\textsc{Op\#}_B^A(M)$&
The operator Schmidt rank of \\&$M\in\mathbf{L}(\mathcal{H}_A\otimes\mathcal{H}_B)$.\\
$\textsc{KC\#}(U)$&
The Kraus-Cirac number of \\&a two qubit unitary operator.\\
\end{tabular}
\end{table}

\subsection{The Schmidt decomposition and rank}
For any vector $\ket{\psi}_{AB}$ in a Hilbert space $\mathcal{H}_A\otimes\mathcal{H}_B$, there exist a set of orthonormal vectors $\{\ket{i}_{A}\in\mathcal{H}_A\}_i$ and a set of orthonormal vectors $\{\ket{i}_{B}\in\mathcal{H}_B\}_i$ such that
\begin{equation}
\ket{\psi}_{AB}=\sum_i \lambda_i \ket{i}_{A}\ket{i}_{B},
\label{SchmidtDecomposition}
\end{equation}
where $\{\lambda_i\}_i$ are non-negative real numbers referred to as {\it Schmidt coefficients}. The decomposition of a vector $\ket{\psi}_{AB}$ given in the form of Eq.~(\ref{SchmidtDecomposition}) is referred to as a {\it Schmidt decomposition} of $\ket{\psi}_{AB}$.    Each Schmidt coefficient is equivalent to the square root of an eigenvalue of the reduced density operator $\rho_A$ on $\mathcal{H}_A$ of $\ket{\psi}_{AB}$ given by
\begin{equation}
\rho_A={\rm tr}_B\left(\ket{\psi}\bra{\psi}_{AB}\right),
\end{equation}
where ${\rm tr}_B(X)=\sum_i\bra{b_i}_B X\ket{b_i}_B$ is a partial trace of $X\in\mathbf{L}(\mathcal{H}_A\otimes\mathcal{H}_B)$ and $\{\ket{b_i}_B\in\mathcal{H}_B\}$ is an orthonormal basis.  The Schmidt decomposition is uniquely determined up to arbitrary choices of the orthonormal vectors in the subspaces corresponding to degenerate Schmidt coefficients.

The number of non-zero coefficients $|\{\lambda_i> 0\}|$ is called as the {\it Schmidt rank} of $\ket{\psi}_{AB}$ .  The Schmidt rank of $\ket{\psi}_{AB}$ is denoted by $\textsc{Sch\#}_B^A(\ket{\psi}_{AB})=|\{\lambda_i> 0\}|$ in this paper.   For a vector $\ket{\psi}_{ABC}$ in a multiple Hilbert space $\mathcal{H}_A\otimes\mathcal{H}_B\otimes\mathcal{H}_C$, the Schmidt decomposition of $\ket{\psi}_{ABC}$  in terms of a bipartite devision between the Hilbert spaces $\mathcal{H}_A\otimes \mathcal{H}_B$ and $\mathcal{H}_C$ can be similarly defined by introducing a set of orthonormal vectors in each devision  and the corresponding Schmidt rank of $\ket{\psi}_{ABC}$ is denoted by $\textsc{Sch\#}_C^{AB}(\ket{\psi}_{ABC})$.

\subsection{The operator Schmidt decomposition and rank}
The {\it operator Schmidt decomposition} can be applied to any linear operators acting on a Hilbert space $\mathcal{H}_A\otimes\mathcal{H}_B$. The set of linear operators $\mathbf{L}(\mathcal{H}_X)$ forms a Hilbert space with respect to the inner product $(M,N)=\frac{1}{\dim(\mathcal{H}_X)}{\rm tr}(M^{\dag}N)$. Thus we can apply the Schmidt decomposition to operators, such that for any linear operator $M\in\mathbf{L}(\mathcal{H}_A\otimes \mathcal{H}_B)$, there exists a set of orthonormal operators $\{P_i\in\mathbf{L}(\mathcal{H}_A)\}_i$ and  $\{Q_i\in\mathbf{L}(\mathcal{H}_B)\}_i$ satisfying
\begin{equation}
M=\sum_i \lambda_i P_i\otimes Q_i,
\label{OperatorSchmidtDecomposition}
\end{equation}
where $\{\lambda_i\}$ are non-negative real numbers referred to as {\it operator Schmidt coefficients} \cite{SchmidtDecomp} of $M$.   The decomposition of a linear operator $M$ given in the form of Eq.~(\ref{OperatorSchmidtDecomposition}) is referred to as the {\it operator Schmidt decomposition} of $M$.   The number of non-zero coefficients $|\{\lambda_i> 0\}|$ is referred to as the {\it operator Schmidt rank} of $M$. In this paper, we denote the operator Schmidt rank of $M$ by $\textsc{Op\#}_B^A(M)$.

\subsection{The Kraus-Cirac decomposition and rank}

A general two-qubit unitary operation $U \in \mathbf{U}(\mathcal{H}_A\otimes\mathcal{H}_B)$ where $\mathcal{H}_A= \mathcal{H}_B= \mathbb{C}^2$  can be decomposed into a canonical form called the Kraus-Cirac decomposition introduced in \cite{KCdec1,KCdec2, KCdec3} given by
\begin{equation}
U=(u \otimes u^\prime) e^{ i(x X \otimes X+y Y \otimes Y+z Z \otimes Z)} (w \otimes w^\prime),
\label{KCD}
\end{equation}
where $u$, $u^\prime$, $w$ and $w^\prime$ are single-qubit unitary operations and $X$, $Y$ and $Z$ are the Pauli operators on $\mathbb{C}^2$ and $x,y,z \in \mathbb{R}$.    In particular, the two-qubit global unitary part $U_{global} (x,y,z) $ of $U$ is defined by
\begin{eqnarray}
U_{global} (x,y,z) := e^{ i(x X \otimes X+y Y \otimes Y+z Z \otimes Z)}.
\label{gloablpartofKCD}
\end{eqnarray}
In Eq.~(\ref{gloablpartofKCD}), the parameters $x,y,z$ in $0 \leq x < \pi/2$ (or $0 \leq x \leq \pi/4$ if $z=0$), $0 \leq y \leq \min\{x, \pi/2 - x \}$ and $0 \leq z \leq y $ cover all two-qubit global unitary operations up to the local unitarily equivalence (the Weyl chamber \cite{KCdec3}).

The Kraus-Cirac number of a two-qubit unitary operation $U$ is defined as the number of non-zero parameters $x,y,z$ in $U_{global}(x,y,z)$ and is denoted by \textsc{KC\#}$(U)$ in this paper.  \textsc{KC\#}$(U)$ characterizes nonlocal properties (globalness) of $U$ \cite{SoedaAkibue}.   The following is the list of classifications of two-qubit unitary operations:   
\begin{itemize}
\item $U$ with \textsc{KC\#}$(U)=0$ is a product of local unitary operations and satisfies $\textsc{Op\#}^A_B(U)=1$.
\item $U$ with  \textsc{KC\#}$(U)=1$ is locally unitarily equivalent to a controlled unitary operation and satisfies $\textsc{Op\#}^A_B(U)=2$.
\item $U$ with  \textsc{KC\#}$(U)=2$ is locally unitarily equivalent to a special class of two-qubit unitary operations called a matchgate \cite{Valiant,TerhalDiVincenzo, JozsaMiyake} and satisfies $\textsc{Op\#}^A_B(U)=4$
\item The rest of two-qubit unitary operations including {\it a SWAP operation} have \textsc{KC\#}$(U)=3$ and satisfies $\textsc{Op\#}^A_B(U)=4$.
\end{itemize}

\section{Cluster networks}
We denote the Hilbert space of a set of qubits specified by a set $\mathcal{Q}$  by $\mathcal{H}_{\mathcal{Q}}$  and the Hilbert space corresponding to a qubit $Q_k$ specified by an index $k$ by $\mathcal{H}_{Q_k} =\mathbb{C}^2$.   In our setting where quantum communications are restricted but classical communications are unrestricted, quantum communication of a qubit state between two nodes is replaced by quantum teleportation \cite{teleportation} between two nodes.  Since any direction of classical communications is allowed, quantum communication of a single qubit state can be achieved by sharing a maximally entangled two-qubit state between the nodes and the direction of quantum communication is not limited.  Thus quantum network coding is equivalent to perform \textit{local operations} (at each nodes) \textit{and classical communication} (LOCC) assisted by the {\it resource state} that consists of a set of maximally entangled two-qubit states (the Bell pairs) $\ket{\Phi^{+}}=\frac{1}{\sqrt{2}}(\ket{00}+\ket{11})$ shared between the nodes connected by edges.  

We investigate which unitary operations are implementable by LOCC assisted by the resource state for a given network where nodes are represented by a two-dimensional lattice.  We consider that a node represented by $v_{i,j}$ is on the coordinate of the two-dimensional lattice $(i,j)$ and edges connect nearest neighbor nodes. We call these networks {\it cluster networks}. We first give a formal definition of a cluster network.

\begin{definition}
A  network $G=\{\mathcal{V},\mathcal{E},\mathcal{I},\mathcal{O} \}$ is a $(k,N)$-cluster network if and only if, 
\begin{eqnarray}
\mathcal{V}&=&\{v_{i,j};\,1\leq i\leq k,1\leq j\leq N\}\nonumber\\
\mathcal{I}&=&\{v_{i,1};\,1\leq i\leq k\}\nonumber\\
\mathcal{O}&=&\{v_{i,N};\,1\leq i\leq k\}\nonumber\\
\mathcal{E}&=&\mathcal{S} \cup \mathcal{K}
\end{eqnarray}
where
\begin{eqnarray}
\mathcal{S}= \{(v_{i,j},v_{i+1,j});\,1\leq i\leq k-1,1\leq j\leq N)\}, \nonumber\\
\mathcal{K}= \{(v_{i,j},v_{i,j+1});\,1\leq i\leq k,1\leq j\leq N-1)\}, 
\end{eqnarray}
$k \geq 1$ and  $N \geq 1$.  $\mathcal{V}$ represents the set of all nodes, $\mathcal{I}$ and $\mathcal{O}$ represent $k$ input nodes and $k$ output nodes, respectively.  $\mathcal{E}$ represents the set of all edges and $\mathcal{S}$ and $\mathcal{K}$ represent the sets of vertical and horizontal edges, respectively.
\end{definition}

Next we define the resource state corresponding to the $(k,N)$-cluster network.  We introduce qubits  $S_{i,j}^1$ at node $v_{i,j}$ and $S_{i+1,j}^2$ at node $v_{i+1,j}$ to represent a Bell pair corresponding to an edge $ (v_{i,j},v_{i+1,j}) \in \mathcal{S}$.  Similarly, we introduce qubits specified by $K_{i,j}^1$ at node $v_{i,j}$ and $K_{i,j+1}^2$ at node $v_{i,j+1}$ to represent a Bell pair corresponding to an edge $ (v_{i,j},v_{i,j+1}) \in \mathcal{K}$.   We denote the set of all qubits in the resource state by $\mathcal{R}=\{ S_{i,j}^1, S_{i+1,j}^2|1\leq i \leq k-1, 1\leq j\leq N\}\cup\{K_{i,j}^1, K_{i,j+1}^2| 1\leq i \leq k, 1\leq j\leq N-1 \}$. The resource state $\ket{\Phi}_{\mathcal{R}}$ corresponding to a cluster network is defined by the following.
\begin{definition}
For a given $(k, N)$-cluster network, the resource state $\ket{\Phi}_{\mathcal{R}}\in\mathcal{H}_{\mathcal{R}}$ is defined by
\begin{eqnarray}
\ket{\Phi}_{\mathcal{R}}&=&\otimes_{i=1}^{k-1}\otimes_{j=1}^{N}\ket{\Phi^+}_{S_{i,j}^1,S_{i+1,j}^2}\nonumber\\
&&\otimes_{i=1}^{k}\otimes_{j=1}^{N-1}\ket{\Phi^+}_{K_{i,j}^1,K_{i,j+1}^2}.
\label{resourcestate}
\end{eqnarray}
\end{definition}

For example, the $(3,3)$-cluster network and the corresponding resource state are shown in Fig.~\ref{fig:cluster1}.   Note that the resource state for a cluster network represented by Eq.~(\ref{resourcestate}) is related to but different from the cluster states used in measurement-based quantum computation \cite{MBQC}. While we can convert  the resource state for a cluster network into a cluster state by applying a projective measurement on all qubits at each node, a cluster state cannot be converted to the resource state for the corresponding cluster network by LOCC.  The resource states for a cluster network is also closely related to a valence bond solid state \cite{VBS} introduced in condensed matter physics through the projected entangled pair states (PEPS) \cite{PEPS} representation for the valence bond solid states.  The resource states for cluster networks are equivalent to a special type of resource states consisting of Bell pairs used for representing 2D (square lattice) PEPS.   PEPS can be represented as states probabilistically obtainable by independently performing a linear transformation on each node on a resource state.   In contrast, conditional operations at each node depending on the outcomes of measurements in other nodes are performed in our network coding scheme for quantum computation. 

\begin{figure}
\begin{center}
  \includegraphics[height=.42\textheight]{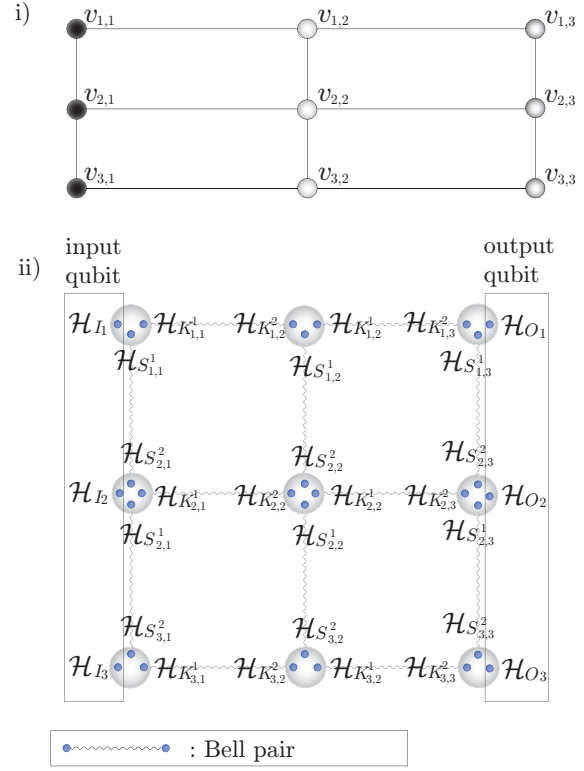}
  \end{center}
\caption{ i) The $(3,3)$-cluster network with the input nodes $\mathcal{I}=\{ v_{1,1},v_{2,1},v_{3,1} \}$, output nodes $\mathcal{O} = \{ v_{1,3},v_{2,3},v_{3,3} \}$ and $3$ repeater nodes $\{ v_{1,2},v_{2,2},v_{3,2} \}$, and ii) the corresponding resource state. Note that the resource states of the cluster networks are different from the cluster states used in measurement-based quantum computation \cite{MBQC}.
}
\label{fig:cluster1}
\end{figure}

Finally we define the implementability of a unitary operation over a $k$-pair network. In addition to resource qubits $\mathcal{R}$,
we introduce input qubits $I_i$ at the input node $v_{i,1} \in \mathcal{I}$, output qubits $O_i$ at the output node $v_{i,N} \in \mathcal{O}$, a set of input qubits $\mathcal{I}_Q =\{I_i|1\leq i\leq k\}$ and a set of output qubits $\mathcal{O}_Q =\{O_i|1\leq i\leq k\}$ for a $(k,N)$-cluster network.
 Note that each input and output node has only one input or output qubit since we concentrate on the implementability of a unitary operation, and the state of output qubits is initially set to be in $\ket{0}\in\mathcal{H}_{\mathcal{O}_Q}$.

\begin{definition}
For a $(k,N)$-cluster network specified by $G=\{\mathcal{V},\mathcal{E},\mathcal{I},\mathcal{O} \}$, a unitary operation $U \in\mathbf{U}(\mathcal{H}_{\mathcal{I}_Q}:\mathcal{H}_{\mathcal{O}_Q})$ is deterministically implementable over the network if and only if there exists a LOCC map $\Gamma$ such that for any pure state $\ket{\psi}\in \mathcal{H}_{ \mathcal{I}_Q}$,
\begin{equation}
\Gamma(\ket{\psi}\bra{\psi}\otimes\ket{\Phi}\bra{\Phi}_{\mathcal{R}})=U\ket{\psi}\bra{\psi}U^{\dag},
\end{equation}
where LOCC map $\Gamma$ consists of local operations on each node and classical communications and $\mathbf{U}(\mathcal{H}_{\mathcal{I}_Q}:\mathcal{H}_{\mathcal{O}_Q})$ is the set of unitary operations from $\mathcal{H}_{\mathcal{I}_Q}$ to $\mathcal{H}_{\mathcal{O}_Q}$.
\end{definition}
Note that the main difference between this network computation model for implementing a unitary operation over a cluster network and standard measurement-based quantum computation is that any operations inside each node are allowed including adding arbitrary ancilla states in this model whereas only projective measurements on the cluster state in each node are allowed in measurement-based quantum computation.  Thus the full set of implementable unitary operations over a ($k,N$)-cluster network is larger than or equal to a set of operations implementable by measurement-based quantum computation using the corresponding cluster states converted from the resource state for the ($k,N$)-cluster network by LOCC.

\section{Conversion of a cluster network into quantum circuits}
We propose a method to convert a $(k,N)$-cluster network into quantum circuits representing a class of unitary operations implementable by LOCC assisted by the resource state corresponding to a given cluster network. By using the converted circuit, it is  easier to construct a network coding protocol since a set of implementable unitary operations are represented by a set of parameters of the converted circuit instead of a complicated LOCC protocol. The class of implementable unitary operations represented by the converted circuit is a subset of that over the cluster network in general since this particular conversion method does not guarantee to give all possible constructions.  However, in some cases,  the constructions given by the conversion methods cover all possible implementable unitary operations as will be shown in Section V.   

We define a set of vertically aligned nodes $\mathcal{V}^v_j := \{ v_{i,j} \}_{i=1}^{k}$ and a set of vertically aligned edges  $\mathcal{S}_j:=\{ (v_{i,j},v_{i+1,j}) \}_{i=1}^{k-1}$ where $1 \leq j \leq N$.   We also define a set of horizontally aligned nodes $\mathcal{V}^h_i := \{ v_{i,j} \}_{j=1}^N$ and a set of horizontally aligned edges  $\mathcal{K}_i:=\{ (v_{i,j},v_{i,j+1}) \}_{j=1}^{N-1}$ where $1 \leq i \leq k$.  We consider that the Bell pairs given for a set of vertically aligned edges  $\mathcal{S}_j$ are used for implementing global unitary operations between nodes whereas each Bell pair given for a set of horizontal aligned edges  $\mathcal{K}_i$ is used for teleporting a qubit state from node $v_{i,j}$ to node $v_{i,j+1}$.

We show that three types of unitary operations, a two-qubit controlled unitary operation, a three-qubit fully controlled unitary operation and a single qubit unitary operation, are implementable between the nodes  in $\mathcal{V}^v_j$ if only one Bell pair is given for each edge and LOCC between the nodes is allowed.  Details of a LOCC protocol implementing three-qubit fully controlled unitary operations are given in the next subsection. A LOCC protocol implementing two-qubit controlled unitary operations has been proposed by \cite{Eisert} and also obtained by simply applying the protocol implementing three-qubit fully controlled unitary operations, as the two-qubit control unitary operations are special cases of three-qubit fully controlled unitary operations.

\begin{itemize}
\item A two-qubit controlled unitary operation:
 A two-qubit controlled unitary operation is defined by
\begin{eqnarray}
C_{ l;n} (\{ u_n^{(a)} \}_{a=0,1} ):=\sum_{a=0}^1 \ket{a}\bra{a}_l \otimes u_n^{(a)},
\end{eqnarray}
where $l$ represents the vertical coordinate of the node $v_{l,j}$ of a control qubit and $n$ represents the vertical coordinate of the node $v_{n,j}$ of a target qubit, and $u_n^{(a)} (a=0,1)$ are single qubit unitary operations on the target qubit.
If $n \neq l \pm 1$,  all Bell pairs represented by edges between  $l$ and $n$ are consumed to implement the two-qubit controlled unitary operation.  When we do not specify the single qubit unitary operations $\{ u_n^{(a)} \}$ on the target qubit we denote a two-qubit
controlled unitary operation simply by $C_{ l;n}$.  In particular, if $u_n^{(0)}=\mathbb{I}_n$ and $u_n^{(1)}=X$,  $C_{ l;n} (\{ u_n^{(a)} \}_{a=0,1} )$ is called as a controlled-NOT operation. 

\item A three-qubit fully controlled unitary operation:
In addition to the two-qubit control unitary operations, we can perform  three-qubit fully controlled unitary operations defined by 
\begin{eqnarray}
C_{l,m;n} (\{ u_n^{(a b)} \}_{a,b=0,1} ):=\sum_{a=0}^1 \sum_{b=0}^1 \ket{a b}\bra{a b}_{lm} \otimes u_n^{(a b)},
\end{eqnarray}
where $l$ and $m$ represent the vertical coordinates of the nodes $v_{l,j}$ and $v_{m,j}$ of two control qubits, respectively,  and $n$  represents the vertical coordinates of the node $v_{n,j}$ of a target qubit, and $u_n^{(a b)} (a,b=0,1)$
represents single qubit operations on the target qubit.
(See the next subsection for details of the LOCC protocol implementing three-qubit fully controlled unitary operations.)
Note that the indices $l$, $n$ and $m$ should be taken such that $l<n<m$ or $m<n<l$. Similarly to the case of a two-qubit controlled unitary operation, we denote a three-qubit fully controlled unitary operation by $C_{ l,m;n} $ when we do not specify the target single qubit operations.  On the other hand, any four-qubit fully controlled unitary, where three of the four qubits are control qubits and the rest of the qubit is a target qubit, is not implementable on  qubits that are all in different nodes of $\mathcal{V}^v_j$ in a $(k,N)$-cluster network, if a single Bell pair is given for each edge in $\mathcal{S}_j$.  

\item A single qubit unitary operation:
Obviously any single qubit unitary operations  can be implemented on any qubit.
\end{itemize}

Note that a general three-qubit fully controlled unitary operation is not implementable by using a sequence of two two-qubit controlled unitary operations that is implementable by using vertically aligned Bell pairs in general. This indicates that the use of three-qubit fully controlled unitary operations enhances the implementability of converted circuits.  A three-qubit fully controlled unitary operation plays an essential role in our network coding protocol over the butterfly network as shown in the next section.

\subsection{A LOCC protocol for implementing three-qubit fully controlled unitary operations}
We show a construction of a protocol to implement a three-qubit fully controlled unitary operation $C_{l,m;n}$ on qubits located at a set of vertically aligned nodes $\mathcal{V}_j^v$ over the $(k,N)$-cluster networks, where $l$ and $m$ represent two control qubits at nodes $v_{l,j}$ and $v_{m,j}$ respectively, and $n$ represents a target qubit at node $v_{n,j}$.    We present a LOCC protocol to implement $C_{l,m;n}$ assisted by the resource states consisting of the Bell pairs corresponding to the vertical edges $\mathcal{S}_j$ of the $(k,N)$-cluster networks.  

We consider to implement $C_{l,m;n}$ on a state of three qubits indexed by $Q_l$, $Q_m$ and $Q_n$ at node $v_{l,j}$, $v_{m,j}$ and $v_{n,j}$, respectively, and its explicit form is given by
\begin{eqnarray}
C_{l,m;n} (\{ u_n^{(ab)} \}):=\sum_{a=0}^1 \sum_{b=0}^1 \ket{ab}\bra{ab}_{lm} \otimes u_n^{(ab)}
\end{eqnarray}
where $\{ \ket{ab} \}_{a,b=\{ 0,1 \} }$ is the two-qubit computational basis of $\mathcal{H}_{Q_l}  \otimes \mathcal{H}_{Q_m}$ of the two controlled qubits and $u_n^{(ab)}$ acts on $\mathcal{H}_{Q_n}$ of the target qubit.  

To show how our LOCC protocol works, we consider an arbitrary state of the control qubits by 
$
\sum \lambda_{ab}\ket{ab}_{lm} \in \mathcal{H}_{Q_l}  \otimes \mathcal{H}_{Q_m} 
$
where
$\{ \lambda_{ab} \}$ is a set of arbitrary complex coefficients satisfying the normalization condition $\sum_{a,b} | \lambda_{ab} |^2 = 1$ and we represent an arbitrary state of the target qubit by $\ket{\phi} \in \mathcal{H}_{Q_n}$.  In the following, we show that our protocol transforms the joint state of controlled qubits and a target qubit to 
\begin{equation}
C_{l,m;n}  \sum_{a,b} \lambda_{ab}\ket{ab}_{lm} \ket{\phi}_n = \sum_{a,b}\lambda_{ab}\ket{ab}_{lm}u_n^{(ab)}\ket{\phi}_n. \nonumber
\end{equation}

The protocol for implementing three qubit fully controlled unitary operations (see Fig.~\ref{fig:3CUcircuit}) is specified as follows:
\begin{enumerate}
\item  Ancillary qubits indexed by $Q_{l^\prime}$, $Q_{m^\prime}$ are introduced at nodes  $v_{l,j}$ and $v_{m,j}$ respectively.  Set both of the ancillary qubits to be in $\ket{0}$.   Each of the two states of control qubits $Q_l$ and $Q_m$ is transformed to a two-qubit state by applying a controlled-NOT operation on the control qubit and the ancillary qubit at the same node, namely applying controlled-NOT operations on $Q_l$ and $Q_{l^\prime}$ and also $Q_m$ and $Q_{m^\prime}$.  Then the joint state of five qubits $Q_l$, $Q_{l^\prime}$, $Q_m$, $Q_{m^\prime}$ and $Q_n$ is given by 
\begin{equation}
\sum_{a,b}\lambda_{ab}\ket{ab}_{lm}\ket{ab}_{l'm'}\ket{\phi}_n.
\end{equation}
\item By consuming the Bell pairs corresponding to the vertical edges $\mathcal{S}_j$ between $v_{l,j}$ and $v_{n,j}$ and also between  $v_{m,j}$ and $v_{n,j}$, perform quantum teleportation to transmit the states of qubits $Q_{l^\prime}$ and $Q_{m^\prime}$ from nodes $v_{l,j}$ and $v_{m,j}$ to $v_{n,j}$. A circuit representation of the protocol of quantum teleportation represented by $\mathcal{T}$ in Fig.~\ref{fig:3CUcircuit} is given by Fig.~\ref{fig:Tcircuit}.  We denote indices of two qubits at node $v_{n,j}$  representing the teleported states from $Q_{l^\prime}$ and $Q_{m^\prime}$ by $Q_{l''}$ and $Q_{m''}$, respectively. 
\item At node $v_{n,j}$, perform $C_{ l,m;n}$ on  $\mathcal{H}_{Q_{l''}}  \otimes \mathcal{H}_{Q_{m''}} \otimes \mathcal{H}_{Q_n}$.  Then we obtain the state given by
\begin{equation}
\sum_{a,b}\lambda_{ab}\ket{ab}_{lm}\ket{ab}_{l''m''}u_n^{(ab)}\ket{\phi}_n.
\end{equation}
\item At node $v_{n,j}$, we apply the Hadamard operations and perform projective measurements in the computational basis on both $\mathcal{H}_{Q_{l''}}$ and $\mathcal{H}_{Q_{m''}}$.   The measurement outcomes of qubits $Q_{l''}$  and $Q_{m''}$ are sent to nodes $v_{l,j}$ and $v_{m,j}$, respectively, by classical communication.   At each of nodes  $v_{l,j}$ and $v_{m,j}$, if the measurement outcome is $0$, do nothing, and if the outcome is $1$, perform $Z$ for a correction on qubit $Q_l$ or $Q_m$.   By straightforward calculation, we obtain the state of three qubits $Q_l$, $Q_m$ and $Q_n$ at nodes $v_{l,j}, v_{m,j}$ and $v_{n,j}$, respectively, given by
\begin{equation}
\sum_{a,b}\lambda_{ab}\ket{ab}_{lm}u_n^{(ab)}\ket{\phi}_n.
\end{equation}
\end{enumerate}

Therefore, $C_{l,m;n}$ is successfully applied on the control qubits at nodes $v_{l,j}$ and $v_{m,j}$ and the target qubit at node $v_{n,j}$ by LOCC assisted by the Bell pairs corresponding to the vertical edges $\mathcal{S}_j$ between nodes $v_{l,j}$ and $v_{m,j}$.

\begin{figure}
\begin{center}
  \includegraphics[height=.25\textheight]{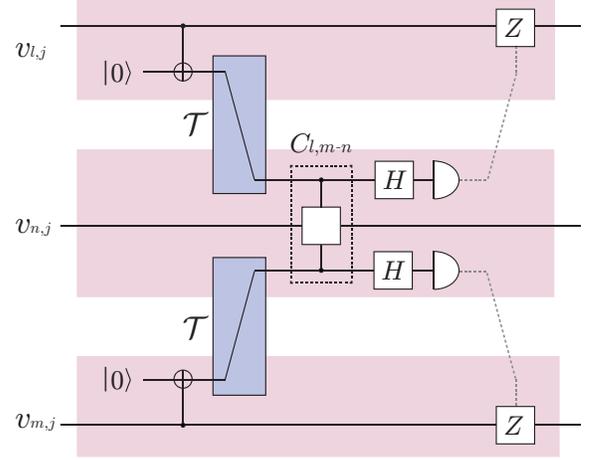}
  \end{center}
  \caption{A quantum circuit representation of the LOCC protocol implementing a three-qubit fully unitary operation $C_{ l,m;n}$, where the qubits in the first shaded region are at the node $v_{l,j}$, those in the second shaded region are at the node $v_{n,j}$ and those in the third shaded region are at the node $v_{m,j}$. The protocol consists of entangling ancillary qubits $Q_{l'}$  and $Q_{m'}$ at the nodes $v_{l,j}$ and $v_{m,j}$, respectively, by performing controlled-NOT operations at $v_{l,j}$ and $v_{m,j}$,  teleporting ancillary qubit states from the nodes $v_{l,j}$ and $v_{m,j}$ to the node $v_{n,j}$ represented by qubits $Q_{l''}$ and $Q_{m''}$ by applying a teleportation protocol denoted by $\mathcal{T}$, applying $C_{ l,m;n}$ on controlled qubits $Q_{l''}$  and $Q_{m''}$ and a target qubit $Q_{n}$ at the node $v_{n,j}$, performing Hadamard operations and measurements in the computational basis on $Q_{l''}$  and $Q_{m''}$ at node $v_{n,j}$ and finally applying conditional $Z$ operations depending on the measurement outcome on two control qubits $Q_l$, $Q_m$ at nodes $v_{l,j}$ and $v_{m,j}$, respectively.}
\label{fig:3CUcircuit}
\end{figure}

\begin{figure}
\begin{center}
  \includegraphics[height=.10\textheight]{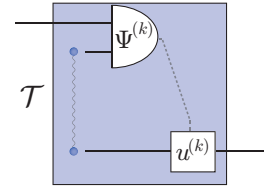}
  \end{center}
  \caption{A quantum circuit representation of the quantum teleportation protocol $\mathcal{T}$, where $\Psi^{(k)}$ represents the measurement projected in the Bell basis (the Bell measurement) $\{\ket{\Psi^{(k)}}\}=\{ (u^{(k)}\otimes \mathbb{I})\ket{\Phi^+} \}$ and $\{u^{(k)}\}=\{I,Z,X,ZX\}$ is a set of operations to be applied conditional on the measurement outcome specified by $k$.  Note that in case of $n \neq l\pm1$, we have to repeat the teleportation protocol to transmit a quantum state between the nodes via the neighboring nodes.  Thus all the Bell pairs corresponding to the vertical edges between $l$ and $n$ are consumed for performing teleportation.}
\label{fig:Tcircuit}
\end{figure}

 \subsection{A conversion protocol}
We present a  protocol to convert a given $(k,N)$-cluster network into quantum circuits. First  (step 1 to step 3), we construct quantum circuits of unitary operations that are implementable on qubits in a set of vertically aligned nodes $\mathcal{V}_j^v$ by LOCC assisted by the Bell pairs given for a set of vertically aligned edges $\mathcal{S}_j$ for a certain $j$.   Then (step 4),  we repeat the procedure given by the first part (step 1 to step 3) for different $j$ of $1 \leq j \leq N$.

The conversion protocol is specified as follows: 
\begin{enumerate}
\item Draw $k$ horizontal wire segments where each of the wire segments corresponds to a set of qubits at vertically aligned nodes $\mathcal{V}^v_j$.

\item Symbols representing two-qubit controlled unitary operations $C_{ l;n}$ or three-qubit fully controlled unitary operations $C_{ l,m;n}$ are added on the horizontal wire segments according to the following rules.

\begin{itemize}
\item To represent  $C_{l;n}$, draw a black dot representing a control qubit on the $l$-th wire, draw a vertical segment from the dot to the $n$-th wire segment and draw a box representing a target unitary operation on the $n$-th wire segment at the end of the vertical segment.  Write index $l$ at the side of the vertical segment between the horizontal wire segments. An example is shown in Fig.~\ref{fig:convertedcircuitex1} i).
\item To represent $C_{ l,m;n}$, draw two black dots representing control qubits on the $l$-th and $m$-th wire segments, draw vertical segments from each dot to the $n$-th wire and draw a box representing an arbitrary target unitary operation on the $n$-th wire segment at the end of the vertical segment.  Write indices $l$ and $m$ at the sides of the vertical segment between the horizontal wire segments. An example is shown in Fig.~\ref{fig:convertedcircuitex1} ii)
\item  Multiple gates of $C_{ l;n}$ or $C_{l,m;n}$ can be added as long as there are only one type of index appearing between the horizontal wire segments and no target unitary operation represented by a box is inserted between two black dots on a horizontal wire segment.
 An example of possible circuits generated in this protocol is shown in Fig.~\ref{fig:convertedcircuitex1} iii). We also give an example of circuits that do not follow the rule in Fig.~\ref{fig:convertedcircuitex1} iv).
\end{itemize}

\item Arbitrary single qubit unitary operations represented by boxes are inserted between before and after the sequence of $C_{ l;n}$ and $C_{ l,m;n}$ (but not during the sequence) obtained by step 2.

\item Repeat step 1 to step 3 for each $1 \leq j \leq N$ and connect all the $i$-th horizontal wire segments.
\end{enumerate}

In Appendix A, we show that a unitary operation represented by the quantum circuit obtained by step 1 to step 3 of the conversion protocol is implementable in a set of vertically aligned nodes $\mathcal{V}_j^v$, namely, it is implementable by LOCC assisted by $(k-1)$ Bell pairs corresponding to a set of vertically aligned edges $\mathcal{S}_j$.  As examples, quantum circuits converted from the $(2,3)$-cluster and $(3,2)$-cluster networks are shown in Fig.~\ref{fig:convertedcircuit}.

\begin{figure}
\begin{center}
  \includegraphics[height=.1\textheight]{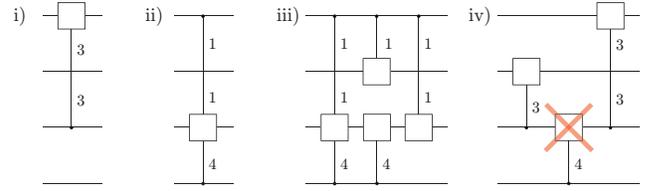}
  \end{center}
 \caption{ i) A symbol representing $C_{3;1}$. ii) A symbol representing $C_{1,4;3}$. iii) An example of circuits generated in step 2 of the conversion protocol. The index in the upper region is 1, that of index in the middle region is 1 and that of index in the lower region is 4.  iv) This conversion is forbidden since there is a target unitary operation inserted between two black dots representing controlled qubits.}
\label{fig:convertedcircuitex1}
\end{figure}

\begin{figure}
\begin{center}
  \includegraphics[height=.35\textheight]{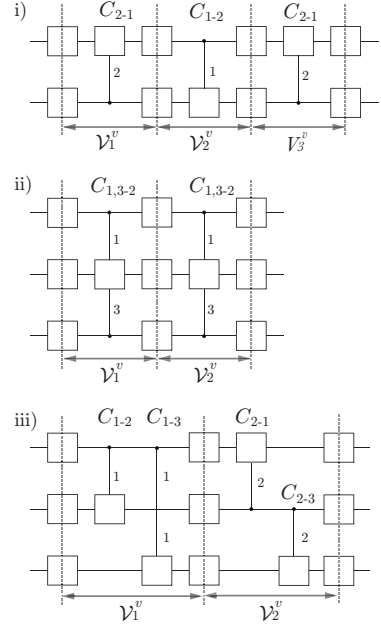}
  \end{center}
  \caption{Arbitrary single qubit unitary operations are represented by boxes. i) An example of converted quantum circuits from the $(2,3)$-cluster network.  It is obtained by connecting three segments of circuits generated in step 1 to  step 3 of the protocol corresponding to $\mathcal{V}_1^v$, $\mathcal{V}_2^v$ and $\mathcal{V}_3^v$. It consists of two-qubit controlled unitary operations defined by $C_{ l;n}=\ket{0}\bra{0}_l\otimes u_n^{(0)}+\ket{0}\bra{0}_l\otimes u_n^{(1)}$, where $l$ denotes the wire segment of the control qubit and $u_n^{(i)}$ are arbitrary single qubit unitary operations on the $n$-th qubit.
  ii) An example of converted quantum circuits from the $(3,2)$-cluster network.  It consists of three-qubit fully controlled unitary operations defined by $C_{ l,m;n}=\ket{00}\bra{00}_{l,m}\otimes u_n^{(00)}+\ket{01}\bra{01}_{l,m}\otimes u_n^{(01)}+\ket{10}\bra{10}_{l,m}\otimes u_n^{(10)}+\ket{11}\bra{11}_{l,m}\otimes u_n^{(11)}$, where $l, m$ denotes the wire segments of the control qubits and $u_n^{(ij)}$ are arbitrary single qubit unitary operations on the $n$-th qubit. iii) Another example of converted quantum circuits obtainable from the $(3,2)$-cluster network.}
\label{fig:convertedcircuit}
\end{figure}

Our conversion method generates infinitely many quantum circuits in general. However for special cluster networks, standard forms of quantum circuits can be obtained.   In Appendix B,  we show that  any converted circuit obtained from a $(2,3)$-cluster network can be simulated by the circuit presented in Fig.~\ref{fig:convertedcircuit} i), 
and any converted circuit obtained from a $(3,2)$-cluster network can be simulated by the circuit presented in Fig.\ref{fig:convertedcircuit} ii).

\section{Implementability of unitary operations over the butterfly and grail networks}

For classical network coding, it has been shown that there exists a network coding protocol over a $2$-pair network, which has two input nodes and two output nodes, if and only if the network has at least one of the butterfly, grail or identity induced substructures \cite{ButterflyGrail1,ButterflyGrail2}.
Thus any classical network coding protocol over a $2$-pair network can be reduced into a combination of protocols over the butterfly, grail or identity networks, and these networks are fundamental primitive networks for classical network coding.   As  a first step to investigate the implementability of quantum computation over general $2$-pair quantum networks, we investigate the implementability of two-qubit unitary operations over the butterfly and grail networks in this section by using the method for converting a $(k, N)$-cluster network into quantum networks introduced in the previous section.

We consider a two-qubit unitary operation $U$ given in the form of the Kraus-Cirac decomposition represented by Eq.~(\ref{KCD}).  Since it is trivial that the single-qubit unitary operations $v$ and $v^\prime$  are implementable at the input nodes and $w$ and $w^\prime$ are implementable at the output nodes, we just need to analyze the implementability of the two-qubit global unitary part $U_{global} (x,y,z) $ given by Eq.~(\ref{gloablpartofKCD}) over the butterfly and grail networks.  Owing to the implementability of a three-qubit fully controlled unitary operation over the butterfly network, we have discovered a protocol for implementing $U_{global}  (x,y,z)$ for arbitrary $x,y,z$ as presented in the constructive proof of the following Theorem.
\begin{figure}
\begin{center}
  \includegraphics[height=.16\textheight]{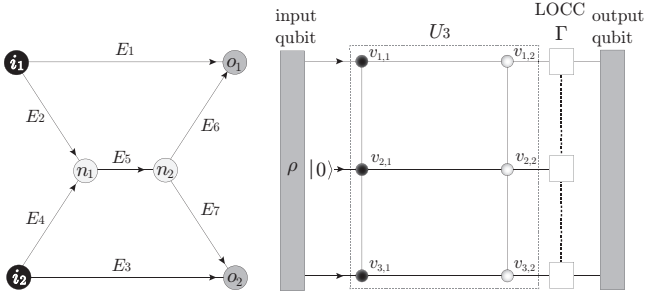}
  \end{center}
  \caption{The nodes $i_1$, $i_2$, $o_1$, $o_2$, $n_1$ and $n_2$ of the butterfly network correspond to the nodes $v_{1,1}$, $v_{3,1}$, $v_{1,2}$, $v_{3,2}$, $v_{2,1}$ and $v_{2,2}$ of a $(3,2)$-cluster network each other. The {\it two}-qubit unitary operation $U_{global} (x,y,z) = e^{ i(x X \otimes X+y Y \otimes Y+z Z \otimes Z)}$ is implementable over a $(3,2)$-cluster network by fixing an input state at the node $v_{2,1}$ at $\ket{0}$, performing an appropriate a {\it three-qubit} unitary operation $U_3$ and performing an appropriate LOCC map $\Gamma$ consisting of a measurement on the qubit at the output node $v_{2,2}$ and the conditional operations on the other output nodes $v_{1,2}$ and $v_{3,2}$ depending on the measurement outcome.
 }
\label{fig:butterflycluster}
\end{figure}

\begin{theorem}
Any two-qubit unitary operation is deterministically implementable over the butterfly network.
\end{theorem}

\begin{proof}
For  the implementability of $U_{global} (x,y,z)$ over the butterfly network represented by the left hand side of Fig.~\ref{fig:butterflycluster}, we consider a  $(3,2)$-cluster network represented by the right hand side of Fig.~\ref{fig:butterflycluster} by assigning the nodes $\{ i_1, n_1, i_2, o_1, n_2, o_2 \}$ of the butterfly network  to the nodes $\{ v_{1,1}, v_{2,1}, v_{3,1}, v_{1,2}, v_{2,2}, v_{3,2} \}$ of the $(3,2)$-cluster network, respectively.   In this assignment, the correspondence of the edges of the butterfly network 
and the horizontal and vertical sets of edges $\mathcal{K}_1, \mathcal{S}_1, \mathcal{S}_2$ of the $(3,2)$-cluster network is given by 
\begin{eqnarray}
\{ E_1, E_5, E_3\} &\leftrightarrow& \mathcal{K}_1, \nonumber \\
\{ E_2, E_4\} &\leftrightarrow& \mathcal{S}_1, \nonumber \\
\{ E_6, E_7\}, &\leftrightarrow& \mathcal{S}_2.
\end{eqnarray}
Thus any two-qubit unitary operation is deterministically implementable over the butterfly network if any $U_{global} (x,y,z) $ in the form of Eq.~(\ref{gloablpartofKCD}) is deterministically implementable over the $(3,2)$-cluster network where input states are given at nodes $v_{1,1}$ and $v_{3,1}$ and output states are obtained at nodes $v_{1,2}$ and $v_{3,2}$, since the topology of the butterfly network is the same as that of the $(3,2)$-cluster network.  

We construct a protocol implementing two-qubit unitary $U_{global}  (x,y,z)$ by setting a fixed input state at node $v_{2,1}$ and arbitrary two-qubit input state at nodes $v_{1,1}$ and $v_{3,1}$ as a three-qubit input state at input nodes $\mathcal{I}=\{v_{1,1}, v_{2,1}, v_{3,1} \}$, and implementing a three-qubit unitary operation denoted by $U_3$ over the $(3,2)$-cluster network followed by an LOCC map denoted by $\Gamma$ performed at output nodes $\mathcal{O}=\{v_{1,2}, v_{2,2}, v_{3,2} \}$.  Recall that a unitary operation represented by the quantum circuit shown in Fig.~\ref{fig:convertedcircuit} ii) is implementable over the $(3,2)$-cluster network.    That is, two three-qubit fully controlled unitary operations $C_{1,3;2}$ are implementable, one at nodes $\mathcal{I}$ and another at nodes $\mathcal{O}$.   The following protocol shows that by choosing appropriate parameters for one of the three-qubit fully controlled unitary operations and one of single-qubit local unitary operations in $U_3$, we can implement $U_{global} (x,y,z)$ with arbitrary $x, y, z$.

The protocol for implementing $U_{global} (x,y,z)$: 
\begin{enumerate}
\item An arbitrary two-qubit input state $\rho$ is given for qubits at input nodes $v_{1,1}$ and $v_{3,1}$ and a fixed input state $\ket{0}$ is prepared for the qubit at node $v_{2,1}$.  
\item Implement $U_3$ of which the quantum circuit representation is given by the left shaded part of Fig.~\ref{fig:butterflycircuit} over the $(3,2)$-cluster network.   
\begin{enumerate}
\item All single-qubit unitary operations appearing in the circuit representation of $U_3$
are trivially performed at each node.  
\item The first fully controlled unitary operation implemented at input nodes $\mathcal{I}$ using the Bell pairs represented by  vertical edges $\mathcal{S}_1$ is given by  $C_{1,3;2} (\{ u_n^{(a b)} \}_{a,b=0,1} )$ where $u_n^{(00)}=u_n^{(11)}=\mathbb{I}$ and $u_n^{(01)}=u_n^{(10)}=Z$.
\item To transmit a qubit state from input nodes $v_{i,1}$ to output node $v_{i,2}$ for $i=1,2,3$, quantum teleportation is performed for each $i$ by using the  Bell pair represented by a horizontal edge in $\mathcal{K}_1$.
\item  The second fully controlled unitary operation implemented at output nodes $\mathcal{O}$ contains parameters $y$ and $z$ and  is given by
$C'_{1,3;2} (\{ w_n^{(a b)} \}_{a,b=0,1} )$ where
\begin{eqnarray}
w_n^{(00)}=w_n^{(11)}=e^{i(z-y)}\ket{0}\bra{0}-ie^{i(z+y)}\ket{1}\bra{1}, \nonumber \\
w_n^{(01)}=w_n^{(10)}=e^{-i(z-y)}\ket{0}\bra{0}-ie^{-i(z+y)}\ket{1}\bra{1}. \nonumber
\end{eqnarray}
\item After implementing $C'_{1,3;2} (\{ w_n^{(a b)} \}_{a,b=0,1} )$, a single-qubit unitary operation parameterized by $x$ given by
\begin{equation}
u(x) = \frac{1}{\sqrt{2}} \left(\begin{array}{cc} e^{i x} & -i e^{- i x} \\e^{i x} & i e^{- i x}\end{array}\right)
\end{equation}
is performed at node $v_{2,2} \in \mathcal{O}$. 
\end{enumerate}
\item Perform an LOCC map $\Gamma$ at output nodes $\mathcal{O}$ of which the quantum circuit representation is given by the right shaded part of Fig.~\ref{fig:butterflycircuit}.  The map  $\Gamma$ consists of the following three steps.
\begin{enumerate}
\item Perform a projective measurement on the qubit at node $v_{2,2}$  in the computational basis $\{ \ket{0}\bra{0}, \ket{1}\bra{1} \}$.
\item Classically communicate the measurement outcome $k \in \{ 0,1 \}$ from node $v_{2,2}$ to $v_{1,2}$ and also to $v_{3,2}$.
\item If $k=1$, perform a conditional operation $X$ on output qubits at nodes $v_{1,2}$ and $v_{3,2}$, otherwise do nothing.
\end{enumerate}
\end{enumerate}

This protocol maps any input state $\rho$ given at input nodes $v_{1,1}$ and $v_{3,1}$ to   
\begin{eqnarray}
U_{global}  (x,y,z)  \rho U_{global}^{\dag} (x,y,z) =\Gamma(U_3(\rho\otimes\ket{0}\bra{0})U_3^{\dag})
\end{eqnarray}
at output nodes $v_{1,2}$ and $v_{3,2}$ where $\ket{0}$ represents the fixed input state at node $v_{2,1}$.  See Appendix D for details of calculations.   It is straightforward to translate the protocol over the $(3,2)$-cluster network to a protocol to implement $U_{global} (x,y,z) $ over the butterfly network by using the correspondence of vertices and edges.  Thus, $U_{global} (x,y,z) $ is deterministically implementable over the butterfly network.
\end{proof}

\begin{figure}
\begin{center}
  \includegraphics[height=.12\textheight]{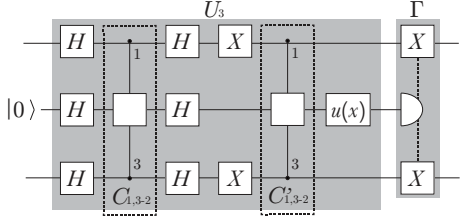}
  \end{center}
  \caption{ A quantum circuit representation of a three-qubit unitary operation $U_3$ (the left shaded part) and an LOCC map $\Gamma$ (the right shaded part) used in a protocol for implementing a two-qubit unitary operation $U_{global} (x,y,z) = e^{ i(x X \otimes X+y Y \otimes Y+z Z \otimes Z)}$ on the first and third qubits.  The input state of the second qubit is fixed in $\ket{0}$.
The single-qubit unitary operations represented by boxes labelled by $H$ and $X$ are given by $H=(\ket{0}\bra{0}+\ket{0}\bra{1} +\ket{1}\bra{0} - \ket{1}\bra{1})/\sqrt{2}$, $u(x)=H(e^{ix}\ket{0}\bra{0}-ie^{-ix}\ket{1}\bra{1}$) and $X= \ket{0}\bra{1} +\ket{1}\bra{0}$, respectively.  The target single-qubit unitary operations of the first three-qubit fully controlled unitary operation $C_{ 1,3;2} (\{ u_n^{(a b)} \}_{a,b=0,1} )$  are given by $u_n^{(00)}=u_n^{(11)}=\mathbb{I}$ and $u_n^{(01)}=u_n^{(10)}=Z$.  The target single-qubit unitary operations of the second three-qubit fully controlled unitary operation $C'_{1,3;2} (\{ w_n^{(a b)} \}_{a,b=0,1} )$  are  given by $w_n^{(00)}=w_n^{(11)}=e^{i(z-y)}\ket{0}\bra{0}-ie^{i(z+y)}\ket{1}\bra{1}$ and $w_n^{(01)}=w_n^{(10)}=e^{-i(z-y)}\ket{0}\bra{0}-ie^{-i(z+y)}\ket{1}\bra{1}$.  The half circle symbol represents a projective measurement in the computational basis $\{ \ket{k} \bra{k} \}_{k=0,1}$.  The single-qubit unitary operations (boxes) connected to the measurement symbol by dotted lines represent conditional unitary operations performed only if the measurement result is $k=1$ and do nothing (or perform $\mathbb{I}$) if $k=0$.}
\label{fig:butterflycircuit}
\end{figure}

In the implementation protocol of $U_{global} (x,y,z)$ presented in the above proof, the first-three qubit fully controlled operation $C_{1,3;2} (\{ u_n^{(a b)} \}_{a,b=0,1} )$ where $u_n^{(00)}=u_n^{(11)}=\mathbb{I}$ and $u_n^{(01)}=u_n^{(10)}=Z$ can be decomposed into a sequence of two controlled-Z operations $C_{1;2} (\{ u_2^{(0)}=\mathbb{I},  u_2^{(1)}=Z\})$ and $C_{3;2} =(\{ u_2^{(0)}=\mathbb{I},  u_2^{(1)}=Z\})$.   This sequence of two-controlled Z operation can be implementable by consuming two Bell pairs corresponding to the edges $E_2$ and $E_4$.   On the other hand, the second three-qubit fully controlled operation  $C'_{1,3;2} (\{ w_n^{(a b)} \}_{a,b=0,1} )$  where $w_n^{(00)}=w_n^{(11)}=e^{i(z-y)}\ket{0}\bra{0}-ie^{i(z+y)}\ket{1}\bra{1}$ and $w_n^{(01)}=w_n^{(10)}=e^{-i(z-y)}\ket{0}\bra{0}-ie^{-i(z+y)}\ket{1}\bra{1}$ {\it cannot} be decomposed into two two-qubit controlled unitary operations which are implementable by using a Bell pair for each.  This is the reason why direct implementability of a three-qubit fully controlled unitary operation by just consuming two vertical Bell pairs corresponding to $E_6$ and $E_7$ is the key for proving implementability of $U_{global} (x,y,z)$ over the butterfly network. 

For the implementability of $U_{global} (x,y,z)$ over the grail network, we consider a $(2,3)$-cluster network by assigning the nodes $\{n_1,n_2,o_1,i_2,n_3,n_4 \}$ of the grail network  to the nodes $\{ v_{1,1}, v_{1,2}, v_{1,3}, v_{2,1}, v_{2,2}, v_{2,3} \}$ of the $(2,3)$-cluster network, respectively (Fig.~\ref{fig:grailcluster}).  The $(2,3)$-cluster network can be converted to a quantum circuit containing three controlled-NOT operations and arbitrary single unitary operations that are inserted between the controlled-NOT operations.   It is shown that any two-qubit unitary operations $U_{global}  (x,y,z)$ can be decomposed by three controlled-NOT gates and single unitary operations inserted between the controlled-NOT operations  \cite{3CNOT}. Thus any two-qubit unitary operation is deterministically implementable over the grail network.

\begin{figure}
\begin{center}
  \includegraphics[height=.16\textheight]{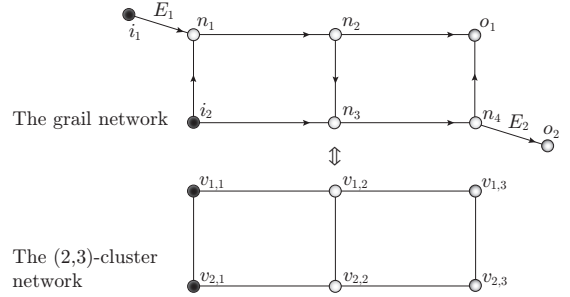}
  \end{center}
  \caption{ The nodes $n_1$, $n_2$, $o_1$, $i_2$, $n_3$ and $n_4$ of the grail network correspond to the nodes $v_{1,1}$, $v_{1,2}$, $v_{1,3}$, $v_{2,1}$, $v_{2,2}$ and $v_{2,3}$ of a $(2,3)$-cluster network, respectively. The set of all unitary operations implementable over the $(2,3)$-cluster network is also implementable over the grail network, since we can use the edges $E_1$ and $E_2$ for just teleporting qubits and the rest of the network forms the $(2,3)$-cluster network, with which any two-qubit unitary operation is implementable.
}  
\label{fig:grailcluster}
\end{figure}

\section{The set of all implementable unitary operations for $k=2,3$}

In this section, we derive the condition for $k$-qubit unitary operations to be implementable over a given cluster network. We show that our conversion method presented in Section III gives all implementable unitary operations over the $(k,N)$-cluster network for $k=2,3$.
\begin{theorem}
If i) a $k$-qubit unitary operation $U$ is deterministically implementable over the $(k,N)$-cluster network ($k\geq 2, N\geq 1$), then ii)  the matrix representation of $U$ in terms of the computational basis $U^M$ can be decomposed into
\begin{equation}
U^M=V_1^MV_2^M\cdots V_N^M,
\label{eq:udec}
\end{equation}
where each $V_i^M$ is a $2^k$ by $2^k$ unitary matrix such that
\begin{eqnarray}
V_i^M&=&\sum_{a_1=0}^1\sum_{a_2=0}^1\cdots\sum_{a_{k-1}=0}^1E^{(a_1)}_{1,i}\otimes E^{(a_1,a_2)}_{2,i}\otimes E^{(a_2,a_3)}_{3,i}\nonumber\\
&&\otimes\cdots\otimes E^{(a_{k-2},a_{k-1})}_{k-1,i}\otimes E^{(a_{k-1})}_{k,i},
\label{eq:dec}
\end{eqnarray}
where $E_{i,j}^{(m,n)}$ and $E_{i,j}^{(m)}$ are 2 by 2 complex matrices.
\end{theorem}
 
To prove Theorem 2, we  first prove Lemma \ref{lemma:subunitary} about a class of bipartite {\it separable maps}  that preserves entanglement.
A bipartite separable map $\Gamma_{sep}$ is a completely positive and trace preserving (CPTP) map whose Kraus operators are product as follows:
\begin{equation}
\Gamma_{sep}(\rho_{EF})=\sum_k (E_k\otimes F_k)\rho_{EF}(E_k\otimes F_k)^{\dag},
\end{equation}
where $\sum_k (E_k\otimes F_k)^{\dag}(E_k\otimes F_k)=I_E\otimes I_F$.
Since quantum network coding is equivalent to perform LOCC assisted by the resource state in our setting, we have to analyze multipartite LOCC. However, the analysis of multipartite LOCC is extremely difficult. Thus, we analyze multipartite separable maps, which are much easier to analyze than LOCC due to their simple structure. Note that a set of separable maps is exactly larger than that of LOCC \cite {sep}.

Let $\ket{\Psi_{in}}=\frac{1}{\sqrt{d}}\sum_{i=0}^{d-1} \ket{A_i}\ket{B_i}$ and $\ket{\Psi_{out}}=\frac{1}{\sqrt{d}}\sum_{i=0}^{d-1} \ket{a_i}\ket{b_i}$, where $\{\ket{A_i}\in\mathcal{H}_A\}$ and $\{\ket{B_i}\in\mathcal{H}_B\}$ are orthonormal sets and $\{\ket{a_i}\in\mathcal{H}_a\}$ and $\{\ket{b_i}\in\mathcal{H}_b\}$ are orthonormal bases. Note that $\dim(\mathcal{H}_a)=\dim(\mathcal{H}_b)=d$ but the dimensions of $\mathcal{H}_A$ and $\mathcal{H}_B$ can be higher than $d$, therefore $\{\ket{A_i} \}$ and $\{\ket{B_i}\}$ do not need to form bases.

\begin{lemma}
Let $\{E_k\in\mathbf{L}(\mathcal{H}_A:\mathcal{H}_a)\},\{F_k\in\mathbf{L}(\mathcal{H}_B:\mathcal{H}_b)\}$ be sets of linear operators. If $\{ E_k \otimes F_k \}$ satisfies
\begin{equation}
\sum_k E_k^{\dag}E_k\otimes F_k^{\dag}F_k=\mathbb{I}_{AB}
\end{equation}
and for all $k$,
\begin{equation}
E_k\otimes F_k \ket{\Psi_{in}}=\sqrt{p_k}\ket{\Psi_{out}}
\end{equation}
is satisfied, then for all $\{k|p_k\neq0\}$,
\begin{equation}
\exists \alpha_k>0,\,\exists U_k^M\in\mathbf{U}(\mathbb{C}^d),\: E^M_k=\alpha_k U^M_k, F^M_k=\frac{\sqrt{p_k}}{\alpha_k}\overline{U_k^M},
\end{equation}
where $E^M_k$ and $F^M_k$ are $d$ by $d$ matrices such that $(E^M_k)_{i,j}=\bra{a_i}E_k\ket{A_j}$, $(F^M_k)_{i,j}=\bra{b_i}F_k\ket{B_j}$, $\mathbf{U}(\mathbb{C}^d)$ is the set of $d$ by $d$ unitary matrices and $\overline{U^M}$ is the complex conjugate of $U^M$.
\label{lemma:subunitary}
\end{lemma}

\begin{proof}
By straightforward calculation, we obtain
\begin{eqnarray}
\forall k,\,E^M_k (F^M_k)^{T}&=&\sqrt{p_k}I_d\nonumber\\
\Rightarrow \forall k\in\{k|p_k\neq0\},\,F^M_k&=&\sqrt{p_k}((E^M_k)^{-1})^{T}
\label{eq:transposeM},
\end{eqnarray}
and
\begin{equation}
\sum_{k} (E^M_k)^{\dag}E^M_k\otimes (F^M_k)^{\dag}F^M_k=I_{d^2}.
\label{eq:closureM}
\end{equation}
By using Eq.~(\ref{eq:transposeM}) and Eq.~(\ref{eq:closureM}), we obtain
\begin{eqnarray}
&&{\rm tr}\left(\sum_k E_k^{M\dag}E_k^M\otimes F_k^{M\dag}F_k^M\right)\nonumber\\
&=&{\rm tr}\left(\sum_{k\in\{k|p_k\neq0\}} E_k^{M\dag}E_k^M\otimes F_k^{M\dag}F_k^M\right)+\epsilon\nonumber\\
&=&\sum_{k\in\{k|p_k\neq0\}}p_k {\rm tr}\left(E_k^{M\dag}E_k^M\otimes \overline{(E_k^{M\dag}E_k^M)^{-1}}\right)+\epsilon=d^2\nonumber\\
&\Leftrightarrow& \sum_{k\in\{k|p_k\neq0\}}p_k {\rm tr}\left(E_k^{M\dag}E_k^M\otimes \overline{(E_k^{M\dag}E_k^M)^{-1}}\right)=d^2-\epsilon,\nonumber\\
\label{eq:20150318}
\end{eqnarray}
where $\epsilon={\rm tr}\left(\sum_{k\in\{k|p_k=0\}} E_k^{M\dag}E_k^M\otimes F_k^{M\dag}F_k^M\right)\geq0$.
We let $P_k=E_k^{M\dag}E_k^M$ be a $d$ by $d$ positive matrix and $\{\lambda_k^i>0|i=0,1,\cdots,d-1\}$ be the set of eigenvalues of $P_k$. Then the eigenvalues of $ \overline{(E_k^{M\dag}E_k^M)^{-1}}=\overline{P_k^{-1}}$ are $\{1/\lambda_k^i|i=0,1,\cdots,d-1\}$ and the condition Eq.~\eqref{eq:20150318} is given by
\begin{equation}
\sum_{k\in\{k|p_k\neq0\}} p_k\sum_{i=0}^{d-1} \lambda_k^i\sum_{j=0}^{d-1} \frac{1}{\lambda_k^j}=d^2-\epsilon.
\label{eq:d-e}
\end{equation}
Using the Cauchy-Schwarz inequality, we obtain
\begin{eqnarray}
\sum_{i=0}^{d-1} \lambda_k^i\sum_{j=0}^{d-1} \frac{1}{\lambda_k^j}&=&\left(\sum_{i=0}^{d-1} \sqrt{\lambda_k^i}^2\right)\left(\sum_{j=0}^{d-1} \sqrt{\frac{1}{\lambda_k^j}}^2\right)\nonumber\\
&\geq& \left(\sum_{i=0}^{d-1}1\right)^2=d^2.
\label{eq:csinequality}
\end{eqnarray}
The equality holds if and only if $\lambda_k^i=\alpha^2 >0$ for all $i$.
By using  Eqs.~(\ref{eq:d-e})-(\ref{eq:csinequality}) and the fact that $\{p_k|p_k\neq0\}$ is a probability distribution, we obtain
for all $k\in\{k|p_k\neq0\}$,
\begin{eqnarray}
\exists \alpha>0;\:P_k=E_k^{M\dag}E_k^M=\alpha^2\mathbb{I}_d,\\
\epsilon=0.
\end{eqnarray}

\end{proof}

{\it Proof of Theorem 2.}
Denote by $\mathcal{H}_{\mathcal{I}_Q}=\otimes_{i=1}^k\mathcal{H}_{I_i}$ and $\mathcal{H}_{\mathcal{O}_Q}=\otimes_{i=1}^k\mathcal{H}_{O_i}$ the Hilbert spaces of $k$ input qubits and $k$ output qubits, respectively.    By introducing another  ancillary Hilbert space $\mathcal{H}_{I'_i}$ at the input nodes $v_{i,1}$, denote the Hilbert space of $k$ qubits by $\mathcal{H}_{\mathcal{I}'_Q}=\otimes_{i=1}^k\mathcal{H}_{I'_i}$. A joint state of $k$ copies of the Bell pairs in $\mathcal{H}_{\mathcal{I}_Q} \otimes \mathcal{H}_{\mathcal{I}'_Q}$ is denoted by 
$$\ket{\mathbb{I}}:= \frac{1}{\sqrt{D}} \sum_i\ket{i}_{\mathcal{I}_Q}\ket{i}_{\mathcal{I}'_Q}=\otimes_{i=1}^k\ket{\Phi^+}_{I_i,I'_i}$$ 
where $D=\dim(\mathcal{H}_{\mathcal{I}_Q})=2^k$.    If $U\in\mathbf{U}(\mathcal{H}_{\mathcal{I}_Q}:\mathcal{H}_{\mathcal{O}_Q})$ is deterministically implementable over a $(k,N)$-cluster network for $k \geq 2$ and $N \geq1$, it is possible to apply $U$ on $\ket{\mathbb{I}}$ and to transmit the resulting state to the output nodes. That is, there exists a LOCC map $\Gamma$ such that
\begin{equation}
\frac{1}{D} \sum_{i,j}\Gamma(\ket{i}\bra{j}_{\mathcal{I}_Q}\otimes\ket{\Phi}\bra{\Phi}_{\mathcal{R}})\otimes \ket{i}\bra{j}_{\mathcal{I}'_Q}=\ket{U}\bra{U},
\label{loccmapfornecessity}
\end{equation}
where $\ket{\Phi}_{\mathcal{R}}$ is the resource state of the $(k,N)$-cluster network and $\ket{U}$ is defined by
\begin{eqnarray}
\ket{U}:=(U \otimes \mathbb{I}) \ket{\mathbb{I}} \in \mathcal{H}_{\mathcal{O}_Q} \otimes \mathcal{H}_{\mathcal{I}'_Q}.
\label{defUket}
\end{eqnarray}

By defining a map represented by the left hand side of Eq.~(\ref{loccmapfornecessity}) as $\Gamma'(\ket{\Phi}\bra{\Phi}_{\mathcal{R}}):= \frac{1}{D} \sum_{i,j}\Gamma(\ket{i}\bra{j}_{\mathcal{I}_Q}\otimes\ket{\Phi}\bra{\Phi}_{\mathcal{R}})\otimes \ket{i}\bra{j}_{\mathcal{I}'_Q}$, where  $\Gamma^\prime$ is also a LOCC map if we assume two qubits belonging to $\mathcal{H}_{I_i}$ and $\mathcal{H}_{I'_i}$ are in the same input node for all $i$.  Since any LOCC maps are separable maps, there exists a separable map $\Gamma^\prime_{sep}$ satisfying
\begin{equation}
\Gamma^\prime_{sep} (\ket{\Phi}\bra{\Phi}_{\mathcal{R}})=\ket{U}\bra{U},
\label{eq:separable}
\end{equation}
if $U$ is deterministically implementable over a $(k,N)$-cluster network.
Since $\Gamma^\prime_{sep}$ is a map from a pure state to a pure state, the action of $\Gamma^\prime_{sep}$ represented by Eq.\eqref{eq:separable} can be equivalently given by the existence of a set of linear operators (the Kraus operators) $\{A_{i,j}^m\}_m$ for each node $v_{i,j}$ and a probability distribution $\{p_m\}$ 
such that
\begin{eqnarray}
\label{eq:LOCC1}
\forall m;\,\, \otimes_{i=1}^k\otimes_{j=1}^N A_{i,j}^m\ket{\Phi}_{\mathcal{R}}&=&\sqrt{p_m}\ket{U},\\
\sum_m \otimes_{i=1}^k\otimes_{j=1}^N (A_{i,j}^{m\dag}A_{i,j}^m)&=&\mathbb{I},
\end{eqnarray}
where
\begin{eqnarray}
A_{i,1}^m&\in&\mathbf{L}(\mathcal{H}_{v_{i,1}}:\mathcal{H}_{I'_i})\nonumber\\
&&(1\leq i\leq k),\nonumber\\
A_{i,j}^m&\in&\mathbf{L}(\mathcal{H}_{v_{i,j}}:\mathbb{C})\nonumber\\
&&(1\leq i\leq k,2\leq j\leq N-1),\nonumber\\
A_{i,N}^m&\in&\mathbf{L}(\mathcal{H}_{v_{i,N}}:\mathcal{H}_{O_i})\nonumber\\
&&(1\leq i\leq k),
\end{eqnarray}
and $\mathcal{H}_{v_{i,j}}$ is the Hilbert space of qubits of the resource state at node $v_{i,j}$ defined by
\begin{equation}
\mathcal{H}_{v_{i,j}}=\mathcal{H}_{S_{i,j}}\otimes\mathcal{H}_{K_{i,j}}
\end{equation}
\begin{subnumcases}
{\mathcal{H}_{S_{i,j}}=}
\mathcal{H}_{S_{1,j}^1} & ($i=1$) \\
\mathcal{H}_{S_{i,j}^1}\otimes\mathcal{H}_{S_{i,j}^2} & ($2\leq i\leq k-1$) \\
\mathcal{H}_{S_{k,j}^2} & ($i=k)$
\end{subnumcases}
\begin{subnumcases}
{\mathcal{H}_{K_{i,j}}=}
\mathcal{H}_{K_{i,1}^1} & ($j=1$)\\
\mathcal{H}_{K_{i,j}^1}\otimes\mathcal{H}_{K_{i,j}^2} & ($2\leq j\leq N-1$) \\
\mathcal{H}_{K_{i,N}^2} & ($j=N)$
\end{subnumcases}

First, letting $E_m=\otimes_{i=1}^k A_{i,1}^m$, $F_m=\otimes_{i=1}^k \otimes_{j=2}^N A_{i,j}^m$ and applying Lemma \ref{lemma:subunitary}, we obtain for all $m\in\{m|p_m\neq 0\}$,
\begin{eqnarray}
\exists \alpha_{1,m}>0,\exists V_{1,m}^M\in\mathbf{U}(\mathbb{C}^{D});\,\,E_m^M=\alpha_{1,m} V_{1,m}^M,
\end{eqnarray}
where $\mathbf{U}(\mathbb{C}^{D})$ is the set of $D$ by $D$ unitary matrices and $E_m^M$ is a $D$ by $D$ matrix satisfying $$(E_m^M)_{a,b}=\bra{a}_{\mathcal{I}'_Q}(\otimes_{i=1}^k A_{i,1}^m)\ket{A_b}_{S_{*,1}^*,K_{*,1}^1}$$ and $$\ket{A_b}_{S_{*,1}^*,K_{*,1}^1}=\otimes_{i=1}^{k-1}\ket{\Phi^+}_{S_{i,1}^1,S_{i+1,1}^2}\otimes \ket{b}_{K_{1,1}^1,\cdots,K_{k,1}^1}.$$
Let
\begin{eqnarray}
A_{1,1}^m&=&\sum_{a_1=0}^1 \bra{a_1}_{S_{1,1}^1}\otimes E_{1,1}^{(a_1),m}\\
A_{i,1}^m&=&\sum_{a_1=0}^1\sum_{a_2=0}^1 \bra{a_1}_{S_{i,1}^1}\bra{a_2}_{S_{i,1}^2}\otimes E_{i,1}^{(a_1,a_2),m}\nonumber\\
&&(2\leq i \leq k-1)\\
A_{k,1}^m&=&\sum_{a_1=0}^1 \bra{a_1}_{S_{k,1}^2}\otimes E_{k,1}^{(a_1),m}
\end{eqnarray}
where $E_{1,1}^{(a_1),m}\in\mathbf{L}(\mathcal{H}_{K_{1,1}^1}:\mathcal{H}_{I'_1})$, $E_{i,1}^{(a_1,a_2),m}\in\mathbf{L}(\mathcal{H}_{K_{i,1}^1}:\mathcal{H}_{I'_i})$ and $E_{k,1}^{(a_1),m}\in\mathbf{L}(\mathcal{H}_{K_{k,1}^1}:\mathcal{H}_{I'_k})$. Thus, $V_{1,m}^M$ can be decomposed into
\begin{eqnarray}
V_{1,m}^M&=&\sum_{a_1,\cdots,a_{k-1}=0}^1 E^{(a_1),m}_{1,1}\otimes E^{(a_1,a_2),m}_{2,1}\otimes\cdots\nonumber\\
&&\otimes E^{(a_{k-2},a_{k-1}),m}_{k-1,1}\otimes E^{(a_{k-1}),m}_{k,1}.
\end{eqnarray}
 Note that we identify a linear operation and its matrix representation in the computational basis, e.g., $E_{1,1}^{(a_1),m}$ is a 2 by 2 complex matrix.

Next, letting $E_m=\otimes_{i=1}^k \otimes_{j=1}^2 A_{i,j}^m$, $F_m=\otimes_{i=1}^k \otimes_{j=3}^N A_{i,j}^m$ and using Lemma \ref{lemma:subunitary}, we obtain for all $m\in\{m|p_m\neq 0\}$,
\begin{equation}
\exists \alpha_{2,m}>0,\exists V_{2,m}^M\in\mathbf{U}(\mathbb{C}^{D});\,\,E_m^M=\alpha_{2,m} V_{2,m}^M,
\end{equation}
where $E_m^M$ is a $D\times D$ matrix such that 
$$(E_m^M)_{a,b}=\bra{a}_{\mathcal{I}'_Q}(\otimes_{i=1}^k\otimes_{j=1}^2 A_{i,j}^m)\ket{A_b}_{S_{*,1}^*,S_{*,2}^*,K_{*,1}^1,K_{*,2}^*}$$ 
and 
\begin{eqnarray}
\ket{A_b}_{S_{*,1}^*,S_{*,2}^*,K_{*,1}^1,K_{*,2}^*}=\otimes_{i=1}^{k-1}\ket{\Phi^+}_{S_{i,1}^1,S_{i+1,1}^2}  \nonumber \\
\otimes_{i=1}^{k-1}\ket{\Phi^+}_{S_{i,2}^1,S_{i+1,2}^2}
\otimes_{i=1}^{k}\ket{\Phi^+}_{K_{i,1}^1,K_{i,2}^2} \nonumber \\
\otimes \ket{b}_{K_{1,2}^1,\cdots,K_{k,2}^1}. \nonumber
\end{eqnarray}
Let
\begin{eqnarray}
A_{1,2}^m&=&\sum_{a_1=0}^1 \bra{a_1}_{S_{1,2}^1}\otimes E_{1,2}^{(a_1),m}\\
A_{i,2}^m&=&\sum_{a_1=0}^1\sum_{a_2=0}^1 \bra{a_1}_{S_{i,2}^1}\bra{a_2}_{S_{i,2}^2}\otimes E_{i,2}^{(a_1,a_2),m}\nonumber\\
&&(2\leq i \leq k-1)\\
A_{k,2}^m&=&\sum_{a_1=0}^1 \bra{a_1}_{S_{k,2}^2}\otimes E_{k,2}^{(a_1),m},
\end{eqnarray}
where $E_{1,2}^{(a_1),m}\in\mathbf{L}(\mathcal{H}_{K_{1,2}^1}\otimes \mathcal{H}_{K_{1,2}^2}:\mathbb{C})$, $E_{i,2}^{(a_1,a_2),m}\in\mathbf{L}(\mathcal{H}_{K_{i,2}^1}\otimes\mathcal{H}_{K_{i,2}^2}:\mathbb{C})$ and $E_{k,2}^{(a_1),m}\in\mathbf{L}(\mathcal{H}_{K_{k,2}^1}\otimes\mathcal{H}_{K_{k,2}^2}:\mathbb{C})$. By straightforward calculation, $V_{2,m}^M$ are shown to be decomposed into
\begin{eqnarray}
V_{2,m}^M&=&V_{1,m}^M\sum_{a_1,\cdots,a_{k-1}=0}^1 E'^{(a_1),m}_{1,2}\otimes E'^{(a_1,a_2),m}_{2,2}\otimes\cdots\nonumber\\
&&\otimes E'^{(a_{k-2},a_{k-1}),m}_{k-1,2}\otimes E'^{(a_{k-1}),m}_{k,2},
\end{eqnarray}
where $E'^{(a_1),m}_{1,2},E'^{(a_1,a_2),m}_{i,2}$, and $E'^{(a_1),m}_{k,2}$ are $2\times 2$ complex matrices.

Iterating this procedure, we obtain for all $m\in\{m|p_m\neq 0\}$,
\begin{equation}
\exists \alpha>0,\exists W^M\in\mathbf{U}(\mathbb{C}^{D});\,\,E_m^M=\alpha W^M,F_m^M=\frac{\sqrt{p_m}}{\alpha}\overline{W^M},
\end{equation}
where $W^M$ and $\overline{W^M}$ can be decomposed into
\begin{eqnarray}
W^M=V_1^MV_2^M\cdots V_{N-1}^M\\
\overline{W^M}=U^{M\dag}V_N^M,
\end{eqnarray}
and $V_i^M=\sum_{a_1,\cdots,a_{k-1}=0}^1E^{(a_1)}_{1,i}\otimes E^{(a_1,a_2)}_{2,i}\otimes\cdots\otimes E^{(a_{k-2},a_{k-1})}_{k-1,i}\otimes E^{(a_{k-1})}_{k,i}\in\mathbf{U}(\mathbb{C}^{D})$.
$U^M$ can be decomposed into the form of Eq.\eqref{eq:udec} since $\overline{V_i^M}$ and $V_i^{M\dag}$ can be decomposed into the form of Eq.\eqref{eq:dec}.
\begin{flushright}
$\Box$
\end{flushright}

In the case of the $(2,N)$-cluster networks, which we call $N$-bridge {\it ladder networks}, $V_i$ is locally unitarily equivalent to the two-qubit controlled unitary operation since its operator Schmidt rank is $2$ \cite{SchmidtDecomp}. Thus,  statements i) and ii) of Theorem 2 are equivalent since a sequence of $N$ two-qubit controlled unitary operations is implementable by the converted circuit presented in Fig.~\ref{fig:convertedcircuit}. Then we obtain the following theorem for the ladder networks.

\begin{theorem}
A unitary operation $U$ is deterministically implementable over the $N$-bridge ladder network if and only if
$\textsc{KC\#}(U)\leq N$.
\label{theorem:main}
\end{theorem}
This theorem is proven by using the following lemma relating the Kraus-Cirac number of a two-qubit unitary operation and the decomposition of the unitary operation into controlled unitary operations  shown in \cite{SoedaAkibue}.

\begin{lemma}
\label{lemma:KC1}
Consider a set of two-qubit unitary operations $\textbf{U}_c$ that is locally unitarily equivalent to a controlled unitary operation.
The decomposition of a unitary operation $U \in SU(4)$ into a shortest sequence of two-qubit unitary operations in $\textbf{U}_c$ depends on the Kraus-Cirac number $\textsc{KC\#}(U)$ of $U$ as
\begin{eqnarray}
\{U\in SU(4)|\textsc{KC\#}(U)\leq1\}&=&\{U|U\in\textbf{U}_c\}\nonumber\\
\{U\in SU(4)|\textsc{KC\#}(U)\leq2\}&=&\{UV|U,V\in\textbf{U}_c\}\nonumber\\
\{U\in SU(4)|\textsc{KC\#}(U)\leq3\}&=&\{UVW|U,V,W\in\textbf{U}_c\}. \nonumber
\end{eqnarray}
\end{lemma}

{\it Proof of Theorem 3.}
Since \textsc{KC\#}$(U)$ is less than or equal to $N$ if and only if $U$ can be decomposed into $N$ two-qubit controlled unitary operations as shown in Lemma \ref{lemma:KC1}, and $N$ two-qubit controlled unitary operations are deterministically implementable over $N$-bridge ladder network, Theorem 3 is straightforwardly shown.
\begin{flushright}
$\Box$
\end{flushright}

We also show that statements  i) and ii) of Theorem 2 are equivalent in the case of  the $(3,N)$-cluster networks in Appendix C. 

\section{Probabilistic implementation of unitary operations}

It is interesting to know whether there exists a task that is not achievable by classical network coding but the corresponding task is achievable in a quantum setting. We can give a negative result in the case of a $(2,2)$-cluster network.
There is no classical network coding protocol to  send single bits from $v_{1,1}$ to $v_{2,2}$ and from $v_{2,1}$ to $v_{1,2}$ over  a $(2,2)$-cluster network since there  is no butterfly, grail or identity substructure. 
This task corresponds to implementing a SWAP operation in quantum network coding.  
  By using Theorem 3, we see that a SWAP operation is not deterministically implementable over  a $(2,2)$-cluster network, which is a 2-bridge ladder network, since the Kraus-Cirac number of the SWAP operation is $3$.  
    
Unitary operations are deterministic maps by definition, but we consider the less restricted situation where the action of the unitary operations are implemented only when we can post-select the preferable probabilistic event.  This corresponds to requiring the  implementing the action of a unitary operation only when certain measurement outcomes in a LOCC protocol are probabilistically obtained.  A formal definition of the probabilistic implementation of a unitary operation is given by Definition 3 by changing LOCC to stochastic LOCC (SLOCC).  In this section, we first characterize all the unitary operations that are probabilistically implementable over cluster networks. Then, we show that a SWAP operation is not implementable even probabilistically.

\begin{theorem}
A $k$-qubit unitary operation $U$ is probabilistically implementable over the $(k,N)$-cluster network ($k\geq 2, N\geq 1$) if and only if the matrix representation of $U$ in terms of the computational basis $U^M$ can be decomposed into
\begin{equation}
U^M=F_1^MF_2^M\cdots F_{N}^M,
\label{eq:Sunitary}
\end{equation}
where each $F_i^M$ is a $2^k$ by $2^k$ complex matrix that can be decomposed in the same way as Eq.~(\ref{eq:dec})
\end{theorem}

\begin{proof}
Similar to the case of deterministic implementation, we consider applying $U \in\mathbf{U}(\mathcal{H}_{\mathcal{I}_Q}:\mathcal{H}_{\mathcal{O}_Q})$ on a part of $k$ maximally entangled states $\ket{\mathbb{I}} \in \mathcal{H}_{\mathcal{I}_Q} \otimes \mathcal{H}_{\mathcal{I}'_Q}$.  Then $U$ is probabilistically implementable over the $(k,N)$-cluster network ($k\geq 2, N\geq 1$) if and only if there exist a stochastic LOCC (SLOCC) map $\Gamma''$ and non-zero probability $p>0$ such that
\begin{equation}
\Gamma''(\ket{\Phi}\bra{\Phi}_{\mathcal{R}})=p\ket{U}\bra{U},
\label{eq:SLOCC}
\end{equation}
where $\ket{\Phi}_{\mathcal{R}}$ is the resource state of the $(k,N)$-cluster network and $\ket{U} \in \mathcal{H}_{\mathcal{O}_Q} \otimes \mathcal{H}_{\mathcal{I}'_Q}$ is defined by Eq.~(\ref{defUket}).
Eq.~\eqref{eq:SLOCC} is equivalent to the statement that there exist a set of  linear operators $\{A_{i,j}\}$ and non-zero probability $p>0$ such that
\begin{equation}
\label{eq:SLOCC1}
\otimes_{i=1}^k\otimes_{j=1}^N A_{i,j}\ket{\Phi}_{\mathcal{R}}=\sqrt{p}\ket{U}.
\end{equation}
The conditions of $\{A_{i,j}\}$ given by Eq.~\eqref{eq:SLOCC1} is similar to the conditions of Kraus operators $\{A_{i,j}^m\}_m$ given by Eq.~\eqref{eq:LOCC1} presented in the proof of Theorem 2. The index $m$ is dropped in Eq.~\eqref{eq:SLOCC1} since the map we consider is SLOCC instead of LOCC considered in Theorem 2. By taking the correspondence between $A_{i,j}$ and $A_{i,j}^m$, we obtain a decomposition of the form presented in Eq.~\eqref{eq:Sunitary}.

\end{proof}

\begin{lemma}
A SWAP operation $U_{swap}  := |00\rangle\langle 00|+|01\rangle\langle 10|+|10\rangle\langle 01|+|11\rangle\langle 11|$ is not probabilistically implementable over the $2$-bridge ladder network.
\end{lemma}
\begin{proof}
By using Theorem 4, the SWAP operation is probabilistically implementable over the $(2,2)$-cluster network ($2$-bridge ladder network) if and only if there exist linear operations $P,Q\in\mathbf{L}(\mathcal{H}_1\otimes\mathcal{H}_2)$ and $E_{i,j}^{(k)}\in\mathbf{L}(\mathcal{H}_i)$ satisfying
\begin{eqnarray}
U_{swap}&=&PQ,\\
P&=&E_{1,1}^{(0)}\otimes E_{2,1}^{(0)}+E_{1,1}^{(1)}\otimes E_{2,1}^{(1)},\label{eq:opsch1}\\
Q&=&E_{1,2}^{(0)}\otimes E_{2,2}^{(0)}+E_{1,2}^{(1)}\otimes E_{2,2}^{(1)},\label{eq:opsch2}
\end{eqnarray}
where $\mathcal{H}_i=\mathbb{C}^2$.
Since $P$ and $Q$ can be decomposed into Eq.\eqref{eq:opsch1} and Eq.\eqref{eq:opsch2}, we can derive
\begin{eqnarray}
\textsc{Op\#}_1^2(P)&\leq& 2,\\
\textsc{Op\#}_1^2(Q)&\leq& 2.
\end{eqnarray}
Since $\textsc{Op\#}_1^2 (U_{swap})=4$, $\textsc{Op\#}_1^2(P)=\textsc{Op\#}_1^2(Q)=2$. In \cite{inverseSchmidtrank}, it is shown that if $\textsc{Op\#}_1^2(P)=2$ and $P$ is invertible, $\textsc{Op\#}_1^2(P^{-1})=2$. Thus, the SWAP operation is probabilistically implementable if and only if there exist linear operations $P,Q\in\mathbf{L}(\mathcal{H}_1\otimes\mathcal{H}_2)$ satisfying
\begin{eqnarray}
Q=U_{swap}P,\\
\textsc{Op\#}_1^2(P)= 2,\,\,{\rm rank}(P)=4\\
\textsc{Op\#}_1^2(Q)= 2,\,\,{\rm rank}(Q)=4.
\end{eqnarray}
In general, we can regard $P$ as a matrix representation of a four qubit pure state $\ket{\Phi}_{1,2,3,4}$;
\begin{equation}
P=\sum_{i=1}^4\bra{i}_{3,4}\ket{\Phi}_{1,2,3,4}\bra{i}_{1,2}.
\end{equation}
Then, the following correspondences are obtained,
\begin{eqnarray}
{\rm rank}(P)=4 &\Leftrightarrow&\textsc{Sch\#}_{1,2}^{3,4}(\ket{\Phi})=4,
\label{eq:4qubit1}\\
\textsc{Op\#}_1^2(P)=2 &\Leftrightarrow&\textsc{Sch\#}_{1,3}^{2,4}(\ket{\Phi})=2,
\label{eq:4qubit2}\\
\textsc{Op\#}_1^2(U_{swap}P)=2 &\Leftrightarrow&\textsc{Sch\#}_{1,4}^{2,3}(\ket{\Phi})=2,
\label{eq:4qubit3}
\end{eqnarray}
where $\textsc{Sch\#}_{1,2}^{3,4}(\ket{\Phi})$ is a Schmidt number in terms of a  partition between qubit $1,2$ and qubit $3,4$. We show that there is no four qubit state simultaneously satisfying Eqs. \eqref{eq:4qubit1}, \eqref{eq:4qubit2}, and \eqref{eq:4qubit3} in Appendix E.
\end{proof}

We can apply Theorem 2 and 4 to a slightly extended cluster network, a cluster network with loops.  We show the definition in Appendix F.

\section{Concluding remarks}

We have investigated the implementability of $k$-qubit unitary operations over $(k,N)$-cluster networks to apply the idea of network coding for distributed quantum computation where  the inputs and outputs of quantum computation are given in all separated nodes and quantum communication between nodes is restricted.   We have presented a method to obtain quantum circuit representations of unitary operations implementable over a given cluster network.   For the $(k,N)$-cluster networks of $k=2,3$, we have shown that our method provides all implementable unitary operations over the cluster network.   The proof is based on the existence of the standard form of the converted quantum circuit and the equivalence of a set of unitary operations represented by the standard form and decomposed into the form given by Eq.~\eqref{eq:dec}.  The proof also suggests that statements i) and ii) of Theorem 2 are equivalent for $k=2,3$. For $k\geq4$, whether our method provides all implementable unitary operations or not is still an open problem since the standard form is not known.

As a first step to  finding the fundamental primitive networks of network coding for quantum settings, we have shown that both of the butterfly and grail networks are sufficient resources for implementing arbitrary two-qubit unitary operations, meanwhile the $(2,2)$-cluster network is not sufficient to implement arbitrary two-qubit unitary operations even probabilistically.  To prove this, we have shown necessary and sufficient conditions of probabilistically implementable unitary operations presented in Theorem 4. There are two differences in Theorem 2 and Theorem 4.  First, we have shown that $U$ can be decomposed into a particular form represented by Eq.~\eqref{eq:udec} {\it if} $U$ is deterministically implementable in Theorem 2 and that $U$ can be decomposed into a particular form represented by Eq.~\eqref{eq:Sunitary} {\it if and only if} $U$ is probabilistically implementable in Theorem 4.  Second, each factor $F_i^M$ in Eq.~\eqref{eq:Sunitary} can be a non-unitary complex matrix while each factor $V_i^M$ in Eq.~\eqref{eq:udec} must be a unitary matrix.   The existence of unitary operations only probabilistically implementable (with less than unit probability) is also left as an open question.

\section*{Acknowledgment}
We acknowledge H. Katsura, F. Le Gall and L. Yu for useful comments.  This work is supported by the Project for Developing Innovation Systems of MEXT, Japan and JSPS by KAKENHI (Grant No. 23540463, No. 23240001, No. 26330006, 15H01677, 16H01050).  We also acknowledge the ELC project (Grant-in-Aid for Scientific Research on Innovative Areas MEXT KAKENHI (Grant No. 24106009)) for encouraging this research.

\newpage
\appendices
\section{LOCC implementation of converted quantum circuits}
We have shown a protocol to implement a three-qubit fully controlled unitary operation in a set of vertically aligned nodes $\mathcal{V}_j^v$. In some cases, we can implement more than one three-qubit or two-qubit controlled unitary operations {\it in parallel} using the same resource.  We show how a sequence of controlled unitary operations represented by converted circuits can be implemented by LOCC assisted by the resource state given by a collection of $(k-1)$ Bell pairs corresponding to a set of vertically aligned edges $\mathcal{S}_j$ in this appendix.

We introduce a new notation for controlled unitary operations for simplifying and unifying descriptions of two-qubit and three-qubit controlled unitary operations.
We represent a two-qubit controlled unitary operation that is controlled by the $i$-th qubit and targets  the $j$-th qubit as
\begin{equation}
(i,i;j),
\end{equation}
and a three-qubit fully controlled unitary operation that is controlled by the $i$-th and $j$-th qubit and targets the $k$-th qubit as
\begin{equation}
(i,j;k).
\end{equation}
Note that we represented $(i,i;j)$ as $C_{i;j}$ and $(i,j;k)$ as $C_{i,j;k}$ in the previous sections.
Let $G=\{g_n\}$ be a sequence of controlled unitary operations that is added in step 2 of the conversion protocol.
For example, the converted circuit represented by Fig.~\ref{fig:circuitsam} is described by a sequence
\begin{eqnarray}
g_1&=&(1,1;2)
\label{g1}
\\
g_2&=&(4,4;2)\\
g_3&=&(1,4;2)\\
g_4&=&(4,4;5)\\
g_5&=&(4,4;3)\\
g_6&=&(5,5;6)\\
g_7&=&(4,4;5)
\label{g7}.
\end{eqnarray}

\begin{figure}
\begin{center}
  \includegraphics[height=.18\textheight]{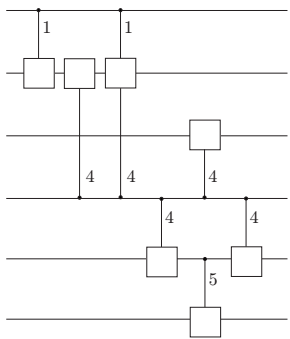}
  \end{center}
  \caption{An example of a converted quantum circuit obtained by step 2 of the conversion protocol.}
\label{fig:circuitsam}
\end{figure}

Let $ \mathcal{C}_i$ be a set of controlled unitary operations that is controlled the $i$-th qubit:
\begin{equation}
 \mathcal{C}_i=\{(a,b;c)\in G;\, a=i\vee b=i\}.
\end{equation}

For example, for $G$ defined by Eqs.~(\ref{g1})-(\ref{g7}),
\begin{eqnarray}
 \mathcal{C}_1=\{g_1,g_3\}
\label{c1}\\
 \mathcal{C}_4=\{g_2,g_3,g_4,g_5,g_7\}\\
\mathcal{C}_5=\{g_6\}\\
 \mathcal{C}_2=C_3=C_6=\emptyset.
\label{c2}
\end{eqnarray}

Define the {\it range} of $ \mathcal{C}_i\neq \emptyset$ as
\begin{eqnarray}
{\rm range}(\mathcal{C}_i)&=&(\min\{i,\min_{c}\{(a,b;c)\in \mathcal{C}_i\}\},\nonumber\\
&& \max\{i,\max_{c}\{(a,b;c)\in  \mathcal{C}_i\}\}).
\end{eqnarray}

For example,  for $\mathcal{C}_i$ defined by Eqs.~(\ref{c1})-(\ref{c2}),
\begin{eqnarray}
{\rm range}(\mathcal{C}_1)&=&(1,2)\\
{\rm range}(\mathcal{C}_4)&=&(2,5)\\
{\rm range}(\mathcal{C}_5)&=&(5,6).
\end{eqnarray}
 
 All the controlled unitary operations in $G$ are implementable by using the following protocol.

The protocol for implementing a sequence of controlled unitary operation in $G$:
\begin{enumerate}
\item For applying gates in $\mathcal{C}_i$, we create an ancillary qubit state entangled to the $i$-th qubit state by preparing an ancillary qubit in $\ket{0}$ and applying a controlled-NOT operation where the ancillary qubit is the target qubit of a controlled-NOT operation.  Then the ancillary qubit state is sent from the $i$-th node $v_{i,j} \in \mathcal{V}_j^v$ to the target node by using teleportation.  If several different target qubits are included in $\mathcal{C}_i$, create another ancillary qubit by the same method at a target node, keep one of the ancillary qubits at the target node and send the other to the next target node.   We consume $n_i$ Bell pairs to teleport ancillary qubit states to the target nodes, where $n_i=b-a$ and ${\rm range}(\mathcal{C}_i)=(a,b)$. Since there is no overlap between ranges of $\mathcal{C}_i$ and there is no target unitary operation inserted  between control qubits, we can teleport all the ancillary qubit states entangled to the control states to all the target nodes by just consuming $(k-1)$ Bell pairs. 
\item
We apply all the controlled unitary operations in $G$ in the target nodes by using the teleported ancillary qubit states entangled to the control qubit states as the control qubits.   
\item
We decouple the ancillary qubit states by performing the projective measurements on the ancillary qubits 
 in the target nodes and apply correction unitary operations in the control nodes depending on the measurement outcomes.
\end{enumerate}

\section{Converted circuit of $(2,N)$ and $(3,N)$-cluster network}

First, we prove that any converted circuits of a $(2,N)$-cluster network can be simulated by a circuit consisting of a sequence of $N$  two-qubit unitary operations and local unitary operations.  In this case, 
only two-qubit unitary operations $(1,1;2)$ or $(2,2;1)$ can be added in step 2 of the conversion protocol. Since applying the gate $(1,1;2)$  sequentially for $k\in\mathbb{N}$ times can be simulated by just one use of gate $(1,1;2)$ and gate $(2,2;1)$ can be simulated by one use of gate $(1,1;2)$ and additional local unitary operations, any circuits generated in step 2 of the conversion protocol can be simulated by one use of $(1,1;2)$ and local unitary operations.

Next, we prove that any converted circuits of  a $(3,N)$-cluster network can be simulated by the circuit of a sequence of $N$  three-qubit fully controlled unitary operations given in the form of
\begin{eqnarray}
\ket{00}\bra{00}_{1,3}\otimes u^{(00)}_2+\ket{01}\bra{01}_{1,3}\otimes u^{(01)}_2\nonumber\\
+\ket{10}\bra{10}_{1,3}\otimes u^{(10)}_2+\ket{11}\bra{11}_{1,3}\otimes u^{(11)}_2
\end{eqnarray}
and local unitary operations.
In step 2 of the conversion protocol, every converted circuits can be simulated by six classes of circuits shown in Fig.~12.

\begin{figure}
\begin{center}
  \includegraphics[height=.23\textheight]{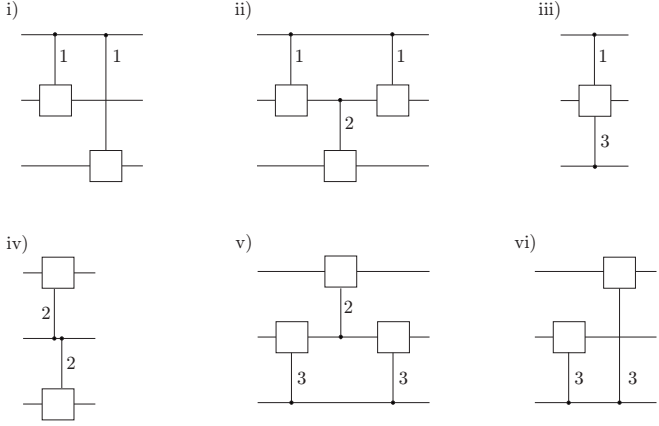}
  \end{center}
  \caption{The six classes of converted quantum circuits obtained by step 2 of the conversion protocol of a $(3,N)$-cluster network.}
\label{fig:convertedcircuit(3,N)}
\end{figure}

In the following, we show that all of these six classes (from class i) to class vi) represented in Fig.~12) can be simulated by a three-qubit fully controlled unitary operation and local unitary operations by investigating each class.
\begin{itemize}
\item[i)] A unitary operation obtained by circuit i) is given by
\begin{eqnarray}
\ket{0}\bra{0}_1\otimes u_2^{(0)}\otimes u_3^{(0)}+\ket{1}\bra{1}_1\otimes u_2^{(1)}\otimes u_3^{(1)}\nonumber\\
\overset{\mathrm{LU}}{=} \ket{0}\bra{0}_1\otimes \mathbb{I}_2\otimes \mathbb{I}_3+\ket{1}\bra{1}_1\otimes u_2^{(1)}u_2^{(0)\dag}\otimes u_3^{(1)}u_3^{(0)\dag}\label{eq:LU1}
\end{eqnarray}
where $u_j^{(i)}$ is a one-qubit unitary operation and $\overset{\mathrm{LU}}{=} $ represents local unitary equivalence. Diagonalize $u_2^{(1)}u_2^{(0)\dag}$ and $u_3^{(1)}u_3^{(0)\dag}$ as
\begin{eqnarray}
u_2^{(1)}u_2^{(0)\dag}&=&v_2 \begin{pmatrix}e^{i\theta_1}&0\\0&e^{i\theta_2}\end{pmatrix}v_2^{\dag}\\
u_3^{(1)}u_3^{(0)\dag}&=&v_3 \begin{pmatrix}e^{i\theta_3}&0\\0&e^{i\theta_4}\end{pmatrix}v_3^{\dag}.
\end{eqnarray}
Since the right-hand side of Eq.~\eqref{eq:LU1} is locally unitarily equivalent to a diagonal unitary operation in  the computational basis, this circuit can be simulated by a three-qubit fully controlled unitary operation and local unitary operations.

\item[ii)]  In circuit ii), the two-qubit controlled unitary operation $(2,2;3)$ can be decomposed into
\begin{eqnarray}
&\ket{0}\bra{0}_2\otimes u_3^{(0)}+\ket{1}\bra{1}_2\otimes u_3^{(1)}\nonumber\\
=&v_3\left(\ket{0}\bra{0}_2\otimes \mathbb{I}_3+\ket{1}\bra{1}_2\otimes \begin{pmatrix}e^{i\theta_1}&0\\0&e^{i\theta_2}\end{pmatrix}\right)v_3^{\dag}u_3^{(0)}\nonumber\\
=&(\mathbb{I}_2\otimes v_3)\nonumber\\
&\left(\begin{pmatrix}1&0\\0&e^{i\theta_1}\end{pmatrix}\otimes\ket{0}\bra{0}_3+\begin{pmatrix}1&0\\0&e^{i\theta_2}\end{pmatrix}\otimes\ket{1}\bra{1}_3\right)\nonumber\\
&(\mathbb{I}_2\otimes v_3^{\dag}u_3^{(0)}),
\end{eqnarray}
where $v_3$ is a unitary operation that diagonalizes $u_3^{(1)}u_3^{(0)\dag}$.
Thus, this circuit is locally unitarily equivalent to a three-qubit fully controlled unitary operation.

\item[iii)] Circuit iii) consists of just a three-qubit fully controlled unitary operation.

\item[iv)] Circuit iv) can be simulated by a three-qubit fully controlled unitary operation and local unitary operations since we can diagonalize a unitary operation obtained by the circuit in the same way as circuit i).

\item[v)] In the same way as circuit ii),  circuit v) is locally unitarily equivalent to a three-qubit fully controlled unitary operation.

\item[vi)] In the same way as circuit i), circuit vi)  is locally unitarily equivalent to a three-qubit fully controlled unitary operation.
\end{itemize}

\section{Two conditions in Theorem 2 are equivalent in the case of the $(3,N)$-cluster networks}

For $k=3$,
the $2^k$ by $2^k$ unitary matrix $V_i^M$ in Theorem 2 is written by
\begin{eqnarray}
V_i^M&=&E_{1,i}^{(0)}\otimes E_{2,i}^{(0,0)}\otimes E_{3,i}^{(0)}\nonumber\\
&&+E_{1,i}^{(0)}\otimes E_{2,i}^{(0,1)}\otimes E_{3,i}^{(1)}\nonumber\\
&&+E_{1,i}^{(1)}\otimes E_{2,i}^{(1,0)}\otimes E_{3,i}^{(0)}\nonumber\\
&&+E_{1,i}^{(1)}\otimes E_{2,i}^{(1,1)}\otimes E_{3,i}^{(1)}.
\end{eqnarray}

By using the result on local unitary equivalence of unitary operations with operator Schmidt rank 2 obtained by Cohen and Yu \cite{Cohen} (Theorem 1 of \cite{Cohen} ), we have
\begin{eqnarray}
V_i^M&\overset{\mathrm{LU}}{=}&\ket{0}\bra{0}_A\otimes W_{BC}^{(0)}+\ket{1}\bra{1}_A\otimes W_{BC}^{(1)}\\
&=&W_{AB}^{(0)}\otimes \ket{0}\bra{0}_C+W_{AB}^{(1)}\otimes\ket{1}\bra{1}_C,
\end{eqnarray}
where $W_{BC}^{(i)}$ and $W_{AB}^{(i)}$ are unitary matrices, $\overset{\mathrm{LU}}{=}$ represents a locally unitarily equivalence and we identify a three-qubit unitary operation on $\mathcal{H}_A\otimes\mathcal{H}_B\otimes\mathcal{H}_C$ as its matrix representation $V_i^M$.
Thus, it is shown that
\begin{eqnarray}
V_i^M&\overset{\mathrm{LU}}{=}&\ket{00}\bra{00}_{AC}\otimes W_{B}^{(00)}+\ket{01}\bra{01}_{AC}\otimes W_{B}^{(01)}\nonumber\\
&&+\ket{10}\bra{10}_{AC}\otimes W_{B}^{(10)}+\ket{11}\bra{11}_{AC}\otimes W_{B}^{(11)},\nonumber\\
\end{eqnarray}
where $W_B^{(ij)}$ is a $2$ by $2$ unitary matrix.
Statements i) and ii) of Theorem 2 of the main text are equivalent in the case of $(3,N)$-cluster networks since $V_i^M$ is a fully controlled  three-qubit unitary operation and $N$ fully controlled three-qubit unitary operations are implementable by a converted circuit of the $(3,N)$-cluster networks.

\section{A network coding protocol for the butterfly network implementing arbitrary two-qubit unitary operations}
We show that the quantum circuit presented in Fig.~\ref{fig:butterflycircuit}
implements a two-qubit global unitary $U_{global}(x,y,z)$ given by Eq.\eqref{gloablpartofKCD}  for arbitrary parameters $x, y, z \in \mathbb{R}$.
$U_{global}(x,y,z)$ can be decomposed into
\begin{eqnarray}
U_{global}(x,y,z)=\sum_j \lambda_j\ket{\Psi^{(j)}}\bra{\Psi^{(j)}}
\end{eqnarray}
by using its eigenvalues $\{\lambda_j\}_j$ and eigenvectors $\{\ket{\Psi^{(j)}}\}_j$ such that
\begin{eqnarray}
\lambda_0=e^{i(x-y+z)},\lambda_1=e^{i(-x+y+z)},\\\lambda_2=e^{i(x+y-z)},\lambda_3=e^{i(-x-y-z)},
\end{eqnarray}
\begin{eqnarray}
\ket{\Psi^{(0)}}&=&\frac{1}{\sqrt{2}}(\ket{00}+\ket{11}),\\
\ket{\Psi^{(1)}}&=&\frac{1}{\sqrt{2}}(\ket{00}-\ket{11}),\\
\ket{\Psi^{(2)}}&=&\frac{1}{\sqrt{2}}(\ket{01}+\ket{10}),\\
\ket{\Psi^{(3)}}&=&\frac{1}{\sqrt{2}}(\ket{01}-\ket{10}).
\end{eqnarray}
Thus, in order to show an arbitrary input state $\ket{\Psi}$ is transformed into $U_{global}(x,y,z)\ket{\Psi}$ through the quantum circuit, it is sufficient to show that the eigenvectors $\{\ket{\Psi^{(j)}}\}_j$ are transformed into $\{\lambda_j\ket{\Psi^{(j)}}\}_j$ and when a measurement is performed, the probability of obtaining a measurement outcome must be independent of the eigenvectors not to break coherence between the eigenvectors.

We divide the quantum circuit into seven steps from step (i) to step (vii) as shown in Fig.~\ref{fig:butterflycircuit1}.  We show the detail of how the eigenvectors are transformed after each step.

First, we prepare a three-qubit input state
\begin{equation}
\ket{\Psi^{(j)}}_{1,3}\ket{0}_2
\end{equation}
in step (i), where we denote the index of the qubit corresponding to the first horizontal wire as 1 and that of the others likewise. After applying Hadamard gates in step (ii), we obtain
\begin{equation}
H_1 H_3\ket{\Psi^{(j)}}_{1,3}\ket{+}_2,
\end{equation}
where $\ket{\pm}=\frac{1}{\sqrt{2}}(\ket{0}\pm\ket{1})$.
After applying $C_{1,3;2}$ in step (iii), we obtain

\begin{equation}
\frac{1}{\sqrt{2}}\left(H_1 H_3\ket{\Psi^{(j)}}_{1,3}\ket{0}_2+Z_1H_1 Z_3H_3\ket{\Psi^{(j)}}_{1,3}\ket{1}_2\right).
\end{equation}
After applying Hadamard gates and Pauli X operations in step (iv), we obtain
\begin{equation}
\frac{1}{\sqrt{2}}\left(X_1 X_3\ket{\Psi^{(j)}}_{1,3}\ket{+}_2+\ket{\Psi^{(j)}}_{1,3}\ket{-}_2\right)=\nonumber
\end{equation}
\begin{subnumcases}
{}
\ket{\Psi^{(j)}}_{1,3}\ket{0}_2 & ($j=0,2$) \\
-\ket{\Psi^{(j)}}_{1,3}\ket{1}_2 & ($j=1,3$).
\end{subnumcases}
After applying $C'_{1,3;2}$ in step (v), we obtain
\begin{subnumcases}
{}
e^{i(-y+z)}\ket{\Psi^{(0)}}_{1,3}\ket{0}_2 & ($j=0$) \\
ie^{i(y+z)}\ket{\Psi^{(1)}}_{1,3}\ket{1}_2 & ($j=1$)\\
e^{i(y-z)}\ket{\Psi^{(2)}}_{1,3}\ket{0}_2 & ($j=2$)\\
ie^{i(-y-z)}\ket{\Psi^{(3)}}_{1,3}\ket{1}_2& ($j=3$).
\end{subnumcases}
After applying a single qubit unitary operation $u(x)$ given by
\begin{equation}
u(x) = \frac{1}{\sqrt{2}} \left(\begin{array}{cc} e^{i x} & -i e^{- i x} \\e^{i x} & i e^{- i x}\end{array}\right)
\end{equation}
in step (vi), we obtain
\begin{subnumcases}
{}
\lambda_j\ket{\Psi^{(j)}}_{1,3}\ket{+}_2 & ($j=0,2$) \\
\lambda_j\ket{\Psi^{(j)}}_{1,3}\ket{-}_2 & ($j=1,3$).
\end{subnumcases}
After applying the projective measurement in the computational basis and conditional unitary operations in step (vii), we obtain 
\begin{equation}
 \lambda_j \ket{\Psi^{(j)}}_{1,3}
\end{equation}
for any measurement outcome. We can verify that the probability of obtaining a measurement outcome is $\frac{1}{2}$ irrespective of eigenvectors.

\begin{figure}
\begin{center}
  \includegraphics[height=.15\textheight]{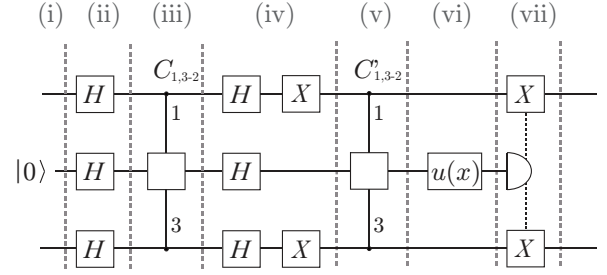}
  \end{center}
  \caption{A protocol to implement a two-qubit unitary operation $U_{global} (x,y,z)$ over the butterfly network. We consider 7 steps  presented in the quantum circuit and denote the steps by Roman numerals, (i) to (vii).  The symbols of gates of the circuit are same as the ones given for Fig.~\ref{fig:butterflycircuit}.}

\label{fig:butterflycircuit1}
\end{figure}

\section{Analysis of a bipartite property of four qubit states}

We prove that there is no pure state of four qubits $\ket{\Phi}_{1,2,3,4}$ satisfying
\begin{eqnarray}
\label{eq:sch1}\textsc{Sch\#}_{1,2}^{3,4}(\ket{\Phi})&=&4,\\
\label{eq:sch2}\textsc{Sch\#}_{2,4}^{1,3}(\ket{\Phi})&=&2,\\
\label{eq:sch3}\textsc{Sch\#}_{2,3}^{1,4}(\ket{\Phi})&=&2.
\end{eqnarray}
In \cite{4qubit}, it is shown that any pure states of four qubits can, up to permutations of the qubits, be transformed into one of the following nine families of states by determinant 1 SLOCC:
\begin{eqnarray}
\ket{\Phi_1}&=&\frac{a+d}{2}(\ket{0000}+\ket{1111})+\frac{a-d}{2}(\ket{0011}+\ket{1100})\nonumber\\
&&+\frac{b+c}{2}(\ket{0101}+\ket{1010})+\frac{b-c}{2}(\ket{0110}+\ket{1001})\nonumber\\
\ket{\Phi_2}&=&\frac{a+b}{2}(\ket{0000}+\ket{1111})+\frac{a-b}{2}(\ket{0011}+\ket{1100})\nonumber\\
&&+c(\ket{0101}+\ket{1010})+\ket{0110}\nonumber\\
\ket{\Phi_3}&=&a(\ket{0000}+\ket{1111})+b(\ket{0101}+\ket{1010})\nonumber\\
&&+\ket{0110}+\ket{0011}\nonumber\\
\ket{\Phi_4}&=&a(\ket{0000}+\ket{1111})+\frac{a+b}{2}(\ket{0101}+\ket{1010})\nonumber\\
&&+\frac{a-b}{2}(\ket{0110}+\ket{1001})\nonumber\\
&&+\frac{i}{\sqrt{2}}(\ket{0001}+\ket{0010}+\ket{0111}+\ket{1011})\nonumber\\
\ket{\Phi_5}&=&a(\ket{0000}+\ket{0101}+\ket{1010}+\ket{1111})\nonumber\\
&&+i\ket{0001}+\ket{0110}-i\ket{1011}\nonumber\\
\ket{\Phi_6}&=&a(\ket{0000}+\ket{1111})+\ket{0011}+\ket{0101}+\ket{0110}\nonumber\\
\ket{\Phi_7}&=&\ket{0000}+\ket{0101}+\ket{1000}+\ket{1110}\nonumber\\
\ket{\Phi_8}&=&\ket{0000}+\ket{1011}+\ket{1101}+\ket{1110}\nonumber\\
\ket{\Phi_9}&=&\ket{0000}+\ket{0111},\nonumber
\end{eqnarray}
where $a,b,c,d$ are complex parameters.

Since the Schmidt number of a state cannot be increased under SLOCC and determinant 1 SLOCC is invertible, the Schmidt number of a state is invariant under determinant 1 SLOCC. Thus, we show that no state of the  nine families simultaneously satisfies Eqs.~(\ref{eq:sch1})-(\ref{eq:sch3}).
There are three ways to divide four qubits into a pair of two qubits. We denote the set of Schmidt numbers of a four qubit state $\ket{\Phi}$ for all bipartite devisions as $\textsc{Sch\#}(\ket{\Phi})=\{\textsc{Sch\#}_{1,2}^{3,4}(\ket{\Phi}),\textsc{Sch\#}_{2,4}^{1,3}(\ket{\Phi}),\textsc{Sch\#}_{2,3}^{1,4}(\ket{\Phi})\}$.

\begin{theorem}
There is no four qubit state $\ket{\Phi}\in\mathcal{H}_1\otimes\mathcal{H}_2\otimes\mathcal{H}_3\otimes\mathcal{H}_4$ such that
\begin{equation}
\textsc{Sch\#}(\ket{\Phi})=\{4,2,2\}.
\label{eq:schcondition}
\end{equation}
\end{theorem}

\begin{proof}
By calculating the Schmidt rank for all bipartite devisions, we can easily check that
\begin{eqnarray}
\textsc{Sch\#}(\ket{\Phi_6})&=&\{n_6,n_6,n_6\}\\
\textsc{Sch\#}(\ket{\Phi_7})&=&\{3,3,3\}\\
\textsc{Sch\#}(\ket{\Phi_8})&=&\{3,3,3\}\\
\textsc{Sch\#}(\ket{\Phi_9})&=&\{2,2,2\},
\end{eqnarray}
where $n_6=\#\left\{\sqrt{2},\frac{1}{2}\sqrt{1+4|a|^2}+\frac{1}{2},\frac{1}{2}\sqrt{1+4|a|^2}-\frac{1}{2}\right\}$ and $\#\mathcal{S}$ is the number of non-zero elements of set $\mathcal{S}$. Since $n_6=2$ or $n_6=3$, these four states do not satisfy Eq.~\eqref{eq:schcondition}.

An element of $\textsc{Sch\#}(\ket{\Phi_5})$ is $\#\left\{1,\sqrt{2},2|a|\right\}$. To satisfy Eq.~\eqref{eq:schcondition}, $a=0$ is required. Then
\begin{equation}
\textsc{Sch\#}(\ket{\Phi_5})=\{2,3,3\},
\end{equation}
which does not satisfy Eq.\eqref{eq:schcondition}.

An element of $\textsc{Sch\#}(\ket{\Phi_4})$ is $\#\{|b|\}+\#\{x|x^3-(3|a|^2+2)x^2+(3|a|^4+2|a|^2+1)x-|a|^6=0\}$.
To satisfy Eq.~\eqref{eq:schcondition}, the element must be 2 or 4.
If the element is 2, since $\#\{x|x^3-(3|a|^2+2)x^2+(3|a|^4+2|a|^2+1)x-|a|^6=0\}$ is larger than 1 and is 2 if and only if $a=0$, we have
\begin{equation}
a=b=0.
\end{equation}
Then $\textsc{Sch\#}(\ket{\Phi_4})=\{2,2,2\}$. Thus, the element must be 4.
Since $\#\{x|x^3-(3|a|^2+2)x^2+(3|a|^4+2|a|^2+1)x-|a|^6=0\}$ is 3 if and only if $a\neq 0$,  we have
\begin{equation}
a\neq0,\,\,b\neq0.
\label{eq:schcon1}
\end{equation}
Another element of $\textsc{Sch\#}(\ket{\Phi_4})$ is $\#\{|a-b|\}+\#\{x|64x^3+(\cdots)x^2+(\cdots)x-|a-b|^4|3a+b|^2=0\}$, where we abbreviate coefficients of $x^2$ and $x$. Since this element must be 2, it is necessary that
\begin{equation}
a-b=0\,\,or\,\,3a+b=0.
\label{eq:schcon2}
\end{equation}
The other element of $\textsc{Sch\#}(\ket{\Phi_4})$ is $\#\{|a+b|\}+\#\{x|64x^3+(\cdots)x^2+(\cdots)x-|a+b|^4|3a-b|^2=0\}$, where we abbreviate coefficients of $x^2$ and $x$. Since this element must be 2, it is necessary that
\begin{equation}
a+b=0\,\,or\,\,3a-b=0.
\label{eq:schcon3}
\end{equation}
We can easily check that it is impossible to simultaneously satisfy Eqs.~(\ref{eq:schcon1})-(\ref{eq:schcon3}).

$\textsc{Sch\#}(\ket{\Phi_3})$ is $\{n_3,n_3',n_3'\}$, where
\begin{eqnarray}
n_3&=&\#\{\sqrt{2},|a+b|,|a-b|\},\\
n_3'&=&\#\{\sqrt{1+4|a|^2}\pm1,\sqrt{1+4|b|^2}\pm1\}.
\end{eqnarray}
To satisfy Eq.~\eqref{eq:schcondition}, $n_3'$ must be $2$, that is $a=b=0$. Then $n_3=1$, which does not satisfy Eq.~\eqref{eq:schcondition}.

$\textsc{Sch\#}(\ket{\Phi_2})$ is $\{n_2,n_2',n_2''\}$, where
\begin{eqnarray}
n_2&=&\#\{|a|,|b|,\sqrt{1+4|c|^2}\pm1\},\\
n_2'&=&\#\{|a+b\pm2c|,\sqrt{1+|a-b|^2}\pm1\},\\
n_2''&=&\#\{|a-b\pm2c|,\sqrt{1+|a+b|^2}\pm1\}.
\end{eqnarray}
In the following, we verify that $\{n_2,n_2',n_2''\}$ cannot be $\{4,2,2\}$, $\{2,4,2\}$ or $\{2,2,4\}$.
\begin{enumerate}
\item $\{n_2,n_2',n_2''\}\neq\{4,2,2\}$:

If $n_2=4$, it is necessary that
\begin{equation}
a\neq0,\,\,b\neq0,\,\,c\neq0.
\label{eq:schcon4}
\end{equation}
If $n_2'=2$, it is necessary that
\begin{eqnarray}
a-b=a+b+2c=0,\\
a-b=a+b-2c=0,\\
or\,\,a+b-2c=a+b+2c=0.
\end{eqnarray}
If $n_2''=2$, it is necessary that
\begin{eqnarray}
a+b=a-b+2c=0,\\
a+b=a-b-2c=0,\\
or\,\,a-b-2c=a-b+2c=0.
\label{eq:schcon5}
\end{eqnarray}
We can easily check that it is impossible to simultaneously satisfy  Eqs.~(\ref{eq:schcon4})-(\ref{eq:schcon5}).

\item $\{n_2,n_2',n_2''\}\neq\{2,4,2\}$:

If $n_2=2$, it is necessary that
\begin{eqnarray}
a=b=0,\\
a=c=0,\\
or\,\,b=c=0.
\end{eqnarray}
With the necessary condition for $n_2''=2$, we obtain that
\begin{equation}
a=b=c=0.
\end{equation}
Then, it is impossible to satisfy $n_2'=4$.

\item $\{n_2,n_2',n_2''\}\neq\{2,2,4\}$:

If $n_2=2$, it is necessary that
\begin{eqnarray}
a=b=0,\\
a=c=0,\\
or\,\,b=c=0.
\end{eqnarray}
With the necessary condition for $n_2'=2$, we obtain that
\begin{equation}
a=b=c=0.
\end{equation}
Then, it is impossible to satisfy $n_2''=4$.
\end{enumerate}

Finally, we analyze $\textsc{Sch\#}(\ket{\Phi_1})$. $\textsc{Sch\#}(\ket{\Phi_1})$ is $\{n_1,n_1',n_1''\}$, where
\begin{eqnarray}
n_1&=&\#\{|a|,|b|,|c|,|d|\},\\
n_1'&=&\#\{|a+b-c-d|,|a-b+c-d|,\nonumber\\
&&|-a+b+c-d|,|a+b+c+d|\},\\
n_1''&=&\#\{|-a+b+c+d|,|a-b+c+d|,\nonumber\\
&&|a+b-c+d|,|a+b+c-d|\}.
\end{eqnarray}
Note that $n_1$, $n_1'$ and $n_1''$ are invariant under permutation of $a$, $b$, $c$ and $d$.
We verify that $\{n_1,n_1',n_1''\}$ cannot be $\{4,2,2\}$, $\{2,4,2\}$ or $\{2,2,4\}$ in the following.
\begin{enumerate}
\item $\{n_1,n_1',n_1''\}\neq\{4,2,2\}$:

If $n_1=4$, it is necessary that
\begin{equation}
a\neq0,\,\,b\neq0,\,\,c\neq0,\,\,d\neq0.
\end{equation}
If $n_1'=2$, it is necessary that in general
\begin{eqnarray}
a+b-c-d=0,\,\,a-b+c-d=0\\
\Leftrightarrow a=d,\,\,b=c.
\end{eqnarray}
Then
\begin{equation}
n_1''=\#\{|2b|, |2a|, |2a|, |2b|\}=4.
\end{equation}

\item $\{n_1,n_1',n_1''\}\neq\{2,4,2\}$ and $\{n_1,n_1',n_1''\}\neq\{2,2,4\}$: 

If $n_1=2$, it is necessary that in general
\begin{equation}
a=0,\,\,b=0,\,\,c\neq0,\,\,d\neq0.
\end{equation}
Then
\begin{equation}
n_1'=n_1''=\#\{|c+d|,|c+d|,|c-d|,|c-d|\}.
\end{equation}

\end{enumerate}

\end{proof}

\section{A cluster network with loops}

A cluster network with loops is defined as follows. 
\begin{definition}
A  network $G=\{\mathcal{V},\mathcal{E},\mathcal{I},\mathcal{O} \}$ is a generalized cluster network if and only if for some $k\geq 1$ and $N\geq 1$, 
\begin{eqnarray}
\mathcal{V}&=&\{v_{i,j};\,1\leq i\leq k,1\leq j\leq N\}\nonumber\\
\mathcal{I}&=&\{v_{i,1};\,1\leq i\leq k\}\nonumber\\
\mathcal{O}&=&\{v_{i,N};\,1\leq i\leq k\}\nonumber\\
\mathcal{E}&=&\mathcal{S}_{sub} \cup \mathcal{K}
\end{eqnarray}
where
\begin{eqnarray}
\mathcal{S}_{sub}&\subseteq&\mathcal{S}_{comp}, \nonumber\\
\mathcal{S}_{comp}&=& \{(v_{m,j},v_{n,j});\,1\leq m<n\leq k,1\leq j\leq N)\}, \nonumber\\
\mathcal{K}&=& \{(v_{i,j},v_{i,j+1});\,1\leq i\leq k,1\leq j\leq N-1)\}.\nonumber\\
\end{eqnarray}
\end{definition}

For this network, if there exists a loop of vertical edges $\mathcal{L}\subseteq \mathcal{S}_{sub}$ such that for some $j$, $L$ and $\{i_m\}_{m=1}^L$,
\begin{eqnarray}
\mathcal{L}&=&\{e_1=(v_{{i_1},j},v_{{i_2},j}),e_2=(v_{{i_2},j},v_{{i_3},j}),\nonumber\\
&&\cdots,e_L=(v_{{i_L},j},v_{{i_1},j})|e_m\neq e_n \,\,{\rm if}\,\, m\neq n\},
\end{eqnarray} 
it allows to perform a cyclic permutation that transmits a qubit state from $v_{i_1,j}$ to $v_{i_2,j}$, from $v_{i_2,j}$ to $v_{i_3,j}$ and so on (by consuming Bell pairs corresponding to the looped vertical edges for teleportation), in addition to performing controlled unitary operations presented in Section IV.    Thus, quantum computation over a cluster network with loops of vertical edges may have more capability than that without a loop.   Note that a condition for the implementable unitary operations over this type of cluster networks with loops are still restricted by Theorem 2 and 4.   An extension of our results for more general networks is an open problem.

\begin{IEEEbiographynophoto}{Seiseki Akibue}
received B.S., M.S. and Ph.D. at the University of Tokyo, Japan in 2011, 2013 and 2016, respectively. Since 2016, he has worked at NTT Communication Science Laboratories, NTT Corporation, Japan. His research interests include foundations of quantum mechanics, distributed quantum computation and quantum computational complexity.
\end{IEEEbiographynophoto}

\begin{IEEEbiographynophoto}{Mio Murao}
received M.S. and Ph.D. at Ochanomizu University in Tokyo, Japan in 1993 and 1996, respectively.  She worked as a postdoctoral fellow at Harvard University (US), Imperial College, London (UK) and RIKEN (Japan).  She was appointed as Associate Professor in 2001 and Professor in 2015 in the Department of Physics, the School of Science, the University of Tokyo. Her research interests cover a wide range of theoretical topics in quantum information and quantum physics.  She currently focuses on investigating entanglement and other non-local properties of quantum mechanics and their applications for distributed quantum information processing.
\end{IEEEbiographynophoto}

\end{document}